\documentclass[11pt,a4paper,reqno]{amsart}

\usepackage[linktocpage]{hyperref}
\hypersetup{
  colorlinks   = true, 
  urlcolor     = blue, 
  linkcolor    = Purple, 
  citecolor   = red 
}
\usepackage{changepage}
\usepackage{amsmath}
\usepackage{amsthm}
\usepackage{amsfonts}
\usepackage{amssymb}
\usepackage{array}
\usepackage{algorithm2e}
\usepackage[margin=2.7cm]{geometry}

\usepackage{mathtools}
\usepackage{soul}
\usepackage{enumitem}
\usepackage{hhline}
\usepackage{multirow} 

\makeatletter
\renewcommand{\tocsection}[3]{%
  \indentlabel{\@ifnotempty{#2}{\bfseries\ignorespaces#1 #2\quad}}\bfseries#3}
\renewcommand{\tocsubsection}[3]{%
  \indentlabel{\@ifnotempty{#2}{\ignorespaces#1 #2\quad}}#3}

\newcommand\@dotsep{4.5}
\def\@tocline#1#2#3#4#5#6#7{\relax
  \ifnum #1>\c@tocdepth 
  \else
    \par \addpenalty\@secpenalty\addvspace{#2}%
    \begingroup \hyphenpenalty\@M
    \@ifempty{#4}{%
      \@tempdima\csname r@tocindent\number#1\endcsname\relax
    }{%
      \@tempdima#4\relax
    }%
    \parindent\z@ \leftskip#3\relax \advance\leftskip\@tempdima\relax
    \rightskip\@pnumwidth plus1em \parfillskip-\@pnumwidth
    #5\leavevmode\hskip-\@tempdima{#6}\nobreak
    \leaders\hbox{$\m@th\mkern \@dotsep mu\hbox{.}\mkern \@dotsep mu$}\hfill
    \nobreak
    \hbox to\@pnumwidth{\@tocpagenum{\ifnum#1=1\bfseries\fi#7}}\par
    \nobreak
    \endgroup
  \fi}
\AtBeginDocument{%
\expandafter\renewcommand\csname r@tocindent0\endcsname{0pt}
}
\def\l@subsection{\@tocline{2}{0pt}{2.5pc}{5pc}{}}
\makeatother

\usepackage{etoolbox}
\makeatletter
\patchcmd{\@setaddresses}{\indent}{\noindent}{}{}
\patchcmd{\@setaddresses}{\indent}{\noindent}{}{}
\patchcmd{\@setaddresses}{\indent}{\noindent}{}{}
\patchcmd{\@setaddresses}{\indent}{\noindent}{}{}
\makeatother

\makeatletter
\newcommand{\mylabel}[2]{#2\def\@currentlabel{#2}\label{#1}}
\makeatother

\newcommand{\C}{\mathcal{C}}

\newcommand{\F}{\mathbb{F}}

\newcommand{\LL}{\mathbb L}
\newcommand{\N}{\mathbb{N}}
\newcommand{\Z}{\mathbb{Z}}

\DeclareMathOperator{\GL}{GL}

\DeclareMathOperator{\Gal}{Gal}
\DeclareMathOperator{\rk}{rk}

\DeclareMathOperator{\End}{End}

\DeclareMathOperator{\Tr}{Tr}

\DeclareMathOperator{\dd}{d}
\DeclareMathOperator{\war}{war}
\DeclareMathOperator{\str}{str}
\DeclareMathOperator{\ssl}{sl}

\DeclareMathOperator{\ess}{ess}

\DeclareMathOperator{\Bil}{Bil}
\DeclareMathOperator{\coeff}{coeff}
\DeclareMathOperator{\Sing}{Sing}

\let\mod\relax\DeclareMathOperator{\mod}{\,mod\,}

\theoremstyle{definition}
\newtheorem{theorem}{Theorem}[section]

\newtheorem{proposition}[theorem]{Proposition}
\newtheorem{corollary}[theorem]{Corollary}
\newtheorem{lemma}[theorem]{Lemma}
\newtheorem{conjecture}[theorem]{Conjecture}

\newtheorem{question}[theorem]{Question}

\newtheorem{definition}[theorem]{Definition}
\newtheorem{example}[theorem]{Example}
\newtheorem{remark}[theorem]{Remark}

\usepackage[dvipsnames,table]{xcolor}
\usepackage{pagecolor}
\definecolor{light-gray}{gray}{0.90}

\usepackage{tikz}
\usetikzlibrary{arrows.meta}
\usepackage{lscape}
\usepackage{arydshln}

\newcommand{\xxi}{x_1,\ldots,x_n}

\newcommand{\DOI}[1]{\href{http://doi.org/#1}{\color{purple}{\tiny\tt DOI:#1}}}
\newcommand{\arxiv}[1]{\href{http://arxiv.org/abs/#1}{{\tiny\tt arXiv:#1}}}

\title{Higher-degree symmetric rank-metric codes}

\author[A. Bik ]{Arthur Bik}
\address{Arthur Bik, \textnormal{Institute for Advanced Study, 1 Einstein Drive, Princeton NJ 08540, USA, and Max-Planck-Institute for Mathematics in the Sciences, Inselstraße 22, 04103 Leipzig, Germany}}
\email{mabik@ias.edu}

\author[A. Neri]{Alessandro Neri}
\address{Alessandro Neri, \textnormal{Department of Mathematics: Analysis, Logic and Discrete Mathematics, Ghent University, Krijgslaan 281, 9000
Gent, Belgium,  and Max-Planck-Institute for Mathematics in the Sciences, Inselstraße 22, 04103 Leipzig, Germany}}
\email{alessandro.neri@ugent.be}

\begin{document}
\begin{abstract}
Over fields of characteristic unequal to $2$, we can identify symmetric matrices with homogeneous polynomials of degree $2$. This allows us to view symmetric rank-metric codes as living inside the space of such polynomials. In this paper, we generalize the construction of symmetric Gabidulin codes to polynomials of degree $d>2$ over field of characteristic $0$ or $>d$. To do so, we equip the space of homogeneous polynomials of degree $d\geq 2$ with the metric induced by the essential rank, which is the minimal number of linear forms needed to express a polynomial. We provide bounds on the minimal distance and dimension of the essential-rank metric codes we construct and provide an efficient decoding algorithm. Finally, we show how essential-rank metric codes can be seen as special instances of rank-metric codes and compare our construction to known rank-metric codes with the same parameters.
\end{abstract}

\maketitle

\section{Introduction}

A rank-metric code $\mathcal{C}$ is a linear space of matrices over a field $\F$ in which the rank of any nonzero matrix is bounded from below be a constant $r>0$, called the minimum distance of the code $\mathcal{C}$. The definition of rank-metric codes goes back to Delsarte \cite{delsarte1978bilinearforms} who studied such codes in connection to the theory of association schemes. Since then, rank-metric codes have appeared in a great variety of contexts. For example, they have found applications as error correcting codes in crisscross error correction \cite{roth1991arraycodes} and random network coding \cite{silva2008randomnetworkcoding}. We refer to \cite{sheekeysurvey} for an introduction to the theory of rank-metric codes and their mathematical insights, and to \cite{bartz2022rank} for a more complete list of their applications.

The most prominent class of of rank-metric codes are the Gabidulin codes \cite{delsarte1978bilinearforms,gabidulin1985maxrankdistance}, whose construction we will repeat here for completeness: Suppose that $\LL/\F$ is a cyclic Galois extension, write $n=[\LL:\F]$ for its degree and let $\sigma$ be a generator of $\Gal(\LL/\F)$. Then, by picking a basis $\alpha_1,\ldots,\alpha_n$ of $\LL$ over $\F$, we can identify the space $\End_{\F}(\LL)$ of $\F$-linear endomorphisms of $\LL$ with the space $\F^{n\times n}$ of $n\times n$ matrices over $\F$. Now, every pair of integers $k,\ell\geq0$ with $k+\ell\leq n-1$ defines a Gabidulin code by considering the space of matrices in $\F^{n\times n}$ corresponding to the subspace
\[
\{c_{-k}\sigma^{-k}+\ldots +c_{\ell}\sigma^{\ell}\mid c_{-k},\ldots,c_{\ell}\in\LL\}\subseteq\End_{\F}(\LL).
\] 
This is a code of dimension $n(k+\ell+1)$ of minimum distance $n-(k+\ell)$. Indeed, an element of the kernel of a nonzero endomorphism $c_{-k}\sigma^{-k}+\ldots +c_{\ell}\sigma^{\ell}$ is also an element of the kernel of $\sigma^{k}(c_{-k})\sigma^0(x)+\ldots +\sigma^{k}(c_{\ell})\sigma^{k+\ell}(x)$.  When $\F=\F_q$ is a finite field and $\sigma$ is the Frobenius automorphism, this is a polynomial of degree $\leq q^{k+\ell}$. The number of roots of this polynomial is bounded from above by its degree, and hence the dimension of the kernel of the endomorphism is bounded above by $k+\ell$. A deeper argument based on Artin's Theorem on the independence of characters shows that the kernel of this endomorphism cannot have dimension greater than $k+\ell$ in general; see e.g. \cite{lam1988vandermonde,roth1996tensorcodes}.
\bigskip

The construction of Gabidulin codes has been generalized in several different directions. Among these, in \cite{roth1996tensorcodes} the construction  has been extended from spaces of matrices to spaces of $d$-way tensors for $d\geq 2$, while in \cite{schmidt2015symmetric} the construction has been modified in such a way that the resulting code consists of only symmetric matrices, that is, the code is a \emph{symmetric rank-metric code}; see also \cite{schmidt2020quadratic}. In this direction, de Boer was the first to study symmetric rank-metric codes, providing constructions of codes with large minimum rank \cite{deboer1996}. His results were also motivated by pure application to classical error-correcting codes. The importance of studying symmetric matrices in the context of rank-metric codes lies also in its applications, due to improved efficiency of decoding algorithms \cite{couvreur2022,gabidulin2004symmetric}.
 Over a field $\F$ of characteristic different from $2$, we can identify symmetric matrices with homogeneous polynomials of degree $2$. This allows us to view symmetric rank-metric codes as living inside the space of such polynomials. More precisely, to a degree $2$ homogeneous polynomial
\[
f=\sum_{1\leq i \leq j \leq n}a_{i,j}x_ix_j \in \F[x_1,\ldots, x_n]_2,
\]
we can associate the matrix
\[
C_f=\begin{pmatrix}
2a_{1,1} & a_{1,2} & \cdots & a_{1,n} \\
a_{1,2} & 2a_{2,2} & \ddots & \vdots \\
\vdots &\ddots & \ddots & a_{n-1,n}\\
a_{1,n} & \cdots &a_{n-1,n} &  2a_{n,n}
\end{pmatrix}.
\]
Viceversa, every symmetric matrix gives rise to a degree-$2$ homogeneous polynomial by reversing the map above, if $\mathrm{char}(\F) \neq 2$.
 
Under this correspondence, the rank of the symmetric matrix $C_f$ can be directly computed from $f$, since it holds that
\begin{equation}\label{eq:rank1}
\rk(C_f)=\min \{r\in \Z_{\geq0} \mid f=g(\ell_1,\ldots,\ell_r), \ell_1,\ldots,\ell_r \in \F[x_1,\ldots,x_n]_1, \ g \in \F[y_1,\ldots,y_r]_2\}.
\end{equation}
In other words, the rank of a symmetric matrix $C_f$ is equal to the minimum number of variables that are needed in order to represent $f$, up to  linear changes of variables. 
This is an equivalent framework in which one can study symmetric rank-metric codes.
\medskip

The goal of this paper is to extend the study of symmetric rank-metric codes to higher degrees of symmetry. In order to do so, we equip the space of degree-$d$ polynomials with a metric induced by the \emph{essential rank}. The essential rank of a degree-$d$ polynomial $f \in \F[x_1,\ldots,x_n]_d$ can be defined as a natural generalization of \eqref{eq:rank1} and is given by 
\[
\ess(f)=\min \{r\in \Z_{\geq0} \mid f=g(\ell_1,\ldots,\ell_r),~ \ell_1,\ldots,\ell_r \in \F[x_1,\ldots,x_n]_1,  g \in \F[y_1,\ldots,y_r]_d\}.
\]
Thus, the essential rank of a polynomial is equal to the minimum number of variables that are needed in order to represent it, up to  linear changes of variables.
Within this framework, we study \emph{essential-rank-metric codes} as subspaces of $\F[x_1,\ldots,x_n]_d$. For every $1\le \rho\le n$, we provide a general construction of a family of essential-rank-metric codes $\mathcal C_{\rho}^{n,d}(\alpha)$ having  essential rank lower bounded by $\rho$; see Theorem \ref{thm:minessrank}. This construction can be seen as a generalization of the family of symmetric Gabidulin codes, as we show that for $d=2$ the constructions of $\mathcal C_{\rho}^{n,2}(\alpha)$ and of symmetric Gabidulin code can be identified with each other. Furthermore, we equip this family of codes with a decoding algorithm, capable to correct all the errors of essential rank at most $\lfloor\frac{\rho-1}{2}\rfloor$; see Algorithm \ref{alg1}. We then explore the connection between essential-rank-metric codes and rank-metric codes, using a result of Carlini \cite{carlini2006reducing}. When the characteristic of the field is larger than $d$ (or equals $0$), essential-rank-metric codes can be seen as special families of rank-metric codes via their first catalecticant matrices; see Corollary \ref{cor:embedding}. Our work has two main underlying motivations: on the one hand, we study the mathematical properties of symmetric rank-metric codes of degree $d$, generalizing classical results for $d=2$; on the other hand we aim to motivate the investigation of metrics induced by other notions of rank for homogeneous polynomials, which are typical of contemporary applied algebraic geometry, such as the \emph{strength} and the \emph{slice rank}.
 
\subsection*{Outline}

 In Section \ref{sec:apolarity}, we discuss the properties of the essential rank and its connection to apolarity theory. Afterwards, in Section~\ref{sec:construction}, we give our construction of a class of essential-rank-metric codes. Here we also provide lower bounds on its dimension and minimum distance. In Section~\ref{sec:decoding}, we give a decoding algorithm that efficiently recovers codewords in our codes that are given up to some bounded error. In Section~\ref{sec:ess-to-rank}, we show how our codes can be seen as a special class of rank-metric codes. For $d=2$, we prove that our construction yields symmetric Gabidulin codes and for $d>2$ we compare the parameters of our codes with the bounds inherited from the theory of rank-metric codes. Finally, in Section~\ref{sec:other_ranks}, we list other notions of rank that could potentially result in metrics that yield useful codes.

\subsection*{Acknowledgements}
AB is partially supported by Postdoc.Mobility Fellowship number P400P2\_199196 from the Swiss National Science
Foundation and a grant from the Simons Foundation (816048, LC). AN is supported by   the FWO (Research Foundation Flanders) grant number 12ZZB23N.

\section{Essential rank of homogeneous polynomials}\label{sec:apolarity}

In this section we briefly recap the  theory of the essential rank of polynomials described by Carlini in \cite{carlini2006reducing}.  Note that in that paper the author assumed that $\mathrm{char}(\F)=0$, even though all results also hold when $\mathrm{char}(\F)>d$. Thus, we will make the assumption that $\mathrm{char}(\F) >d$ or $\mathrm{char}(\F)=0$ for this whole section.

Let $S_n(\F)=\F[\xxi]$ and let $T_n(\F)=\F[\partial_1,\ldots,\partial_n]$ be the ring of differential operators acting on $S_n(\F)$. This means that $S_n(\F)$ is endowed with a $T_n(\F)$-module structure acting via differentiation, that is, the action of $T_n(\F)$ on $S_n(\F)$ is defined via 
\[
\partial_i \circ f\coloneqq \frac{\partial}{\partial x_i}f,\qquad \mbox{ for every } f \in S_n(\F).
\]
We will mostly focus on the homogeneous components of these rings, denoted by $S_{n,d}(\F)$ and $T_{n,d}(\F)$, respectively, where $d$ is the degree. By convention, we will always consider the zero polynomial $0$ to be in $S_{n,d}(\F)$ and the zero operator $0$ to be in $T_{n,d}(\F)$. 
If we restrict to homogeneous components, for all $0\leq b \leq a$, we have a surjective $\F$-bilinear map
\[
\begin{array}{rccc}\Phi_{n,a,b}\colon&S_{n,a}(\F) \times T_{n,b}(\F) & \longrightarrow &S_{n,a-b}(\F) \\
& (f,D) & \longmapsto & D\circ f.\end{array}
\]
In the case $a=b=d$, this gives a natural nondegenerate $\F$-bilinear map 
\begin{equation}\label{eq:pairing} 
\begin{array}{rccc}\Phi_{n,d}\colon&S_{n,d}(\F) \times T_{n,d}(\F) & \longrightarrow &\F \\
& (f,D) & \longmapsto & D\circ f.\end{array}
\end{equation}
Clearly, the bilinear form \eqref{eq:pairing} induces a duality. For given subsets $A\subseteq S_{n,d}(\F)$, $B\subseteq T_{n,d}(\F)$, the orthogonal subspaces
\begin{align*}
A^\perp & := \left\{D \in T_{n,d}(\F) \,\middle|\, D\circ f=0 \mbox{ for all } f \in A \right\} \subseteq T_{n,d}(\F), \\
B^\perp & := \left\{f \in S_{n,d}(\F) \,\middle|\, D\circ f=0 \mbox{ for all } D \in B  \right\} \subseteq S_{n,d}(\F)  
\end{align*}
are such that 
\[
\dim_{\F}(\langle A\rangle_{\F})+\dim_{\F}(A^\perp)=\dim_{\F}(S_{n,d}(\F))=\dim_{\F}(T_{n,d}(\F))=\dim_{\F}(\langle B\rangle_{\F})+\dim_{\F}(B^\perp).
\]
Note that $S_{n,1}(\F)$ is just the space of linear polynomials, while  $T_{n,1}(\F)$ is the space of all directional derivatives.  These spaces are naturally isomorphic to $\F^n$, and for a given vector $v=(v_1,\ldots,v_n) \in \F^n$, we will write 
\begin{align*}v(x) &\coloneqq v_1x_1+\ldots+v_nx_n \in S_{n,1}(\F), \\
v(\partial) &\coloneqq v_1\partial_1+\ldots+v_n\partial_n  \in {T}_{n,1}(\F).
\end{align*}

\begin{definition}
Let $f \in S_{n,d}(\F)$.
The \textbf{essential rank} of $f$ is the integer
\[
\ess(f)\coloneqq \min \{ r \in \Z_{\geq0} \mid f\in \F[\ell_1,\ldots,\ell_r]_d, \ell_1,\ldots,\ell_r \in S_{n,1}(\F)\}.
\]
\end{definition}

The essential rank of $f$ is simply the minimum number of variables needed to express $f$ after a linear change of variables.

\begin{lemma}\label{lem:uniqueness}
Let $f \in S_{n,d}(\F)$ be such that $\ess(f)=r$. Let $\ell_1,\ldots,\ell_r$, $h_1,\ldots,h_r \in S_{n,1}(\F)$ be such that $f\in \F[\ell_1,\ldots,\ell_r]$ and $f\in \F[h_1,\ldots,h_r]$. Then $\langle \ell_1,\ldots,\ell_r\rangle_{\F}=\langle h_1,\ldots,h_r\rangle_{\F}$.
\end{lemma}
\begin{proof}
Assume that $\langle \ell_1,\ldots,\ell_r\rangle_{\F}\neq\langle h_1,\ldots,h_r\rangle_{\F}$ and let $0\leq s<r$ be the dimension of the intersection. By minimality of $r$, both $\ell_1,\ldots,\ell_r$ and $h_1,\ldots,h_r$ are linearly independent sets. After coordinate changes we may assume that $\ell_1=h_1,\ldots,\ell_s=h_s$, so that 
\[
V:=\langle \ell_1,\ldots,\ell_r\rangle_{\F}\cap\langle h_1,\ldots,h_r\rangle_{\F}=\langle \ell_1,\ldots,\ell_s\rangle_{\F}=\langle h_1,\ldots,h_s\rangle_{\F}.
\]
We claim that the forms $\ell_1,\ldots,\ell_r,h_{s+1},\ldots,h_r$ are linearly independent. Indeed, when
\[
c_1\ell_1+\ldots+c_r\ell_r+b_{s+1}h_{s+1}+\ldots+b_rh_r=0,\quad c_1,\ldots,c_r,b_{s+1},\ldots,b_r\in\F,
\]
we see that $b_{s+1}h_{s+1}+\ldots+b_rh_r\in V$ and so $b_{s+1},\ldots,b_r=0$ as $h_1,\ldots,h_r$ are linearly independent. Hence $c_1\ell_1+\ldots+c_r\ell_r=0$ and so $c_1,\ldots,c_r=0$ as $\ell_1,\ldots,\ell_r$ are also linearly independent. This proves the claim. Now, we can write 
\[
f_1(\ell_1,\ldots,\ell_r)=f=f_2(h_1,\ldots,h_r)
\]
and interpret both sides as a polynomial in $\ell_1,\ldots,\ell_r,h_{s+1},\ldots,h_r$. The left-hand side does not depend on the forms $h_{s+1},\ldots,h_r$ and the right-hand side does not depend on $\ell_{s+1},\ldots,\ell_r$. Since both sides are in fact the same polynomial, this means that $f$ is a polynomial in only $\ell_1,\ldots,\ell_s$. This contradicts the minimality of $r$. Hence $\langle \ell_1,\ldots,\ell_r\rangle_{\F}=\langle h_1,\ldots,h_r\rangle_{\F}$.
\end{proof}

\begin{definition}
Let $f\in S_{n,d}(\F)$ with essential rank $\ess(f)=r$. We define the \textbf{space of essential variables} of $f$ to be 
\[
V_{\ess}(f)\coloneqq \langle \ell_1,\ldots,\ell_r \rangle_{\F} \subseteq S_{n,1}(\F),
\]
where $\ell_1,\ldots,\ell_r\in S_{n,1}(\F)$ are such that $f\in\F[\ell_1,\ldots,\ell_r]$.
\end{definition}

\begin{remark}
The proof of Lemma \ref{lem:uniqueness} in fact shows that if $h_1,\ldots,h_s\in S_{n,1}(\F)$ are such that $f\in\F[h_1,\ldots,h_s]$, then $V_{\ess}(f)\subseteq \langle h_1,\ldots,h_s\rangle_{\F}$. 
\end{remark}

\begin{example}\label{exa:1}
Let $\F$ be a field with $\mathrm{char}(\F)\neq 2,3$, and let 
\[
f=3 x_1^3 +  8 x_1^2 x_2 + 5 x_1^2x_3 + 12 x_1 x_2^2 + 4 x_1x_2 x_3 + 4 x_1x_3^2 + 8 x_2^3 + 2 x_2 x_3^2 + x_3^3 \in S_{3,3}(\F).
\]
Even though at a first sight it is not clear what the essential rank of $f$ equals, one can verify that 
\begin{align*}
f&=(x_1+2x_2)(x_1+x_3)^2+(x_1+2x_2)^3+(x_1+x_3)^3\\
&=:g(x_1+2x_2,x_1+x_3).
\end{align*}
This means that $V_{\ess}(f)\subseteq \langle x_1+2x_2,x_1+x_3\rangle_{\F}$ and  $\ess(f)\leq 2$. We will see  in Example \ref{exa:3} that these are in fact equalities.
\end{example}

\begin{remark}\label{rem:basic_properties} Let $f,g \in S_{n,d}(\F)$ and let $\lambda \in \F^*$. The following properties hold.
\begin{enumerate}
\item $\ess(f) \geq 0$, and $\ess(f)=0$ if and only if $f=0$.
\item $\ess(f+g)\leq \ess(f)+\ess(g)$. 
\item $\ess(f)=\ess(\lambda f)$.
\item $\ess(f)\leq n$.   
\item $\ess(f)=1$ if and only if $f=\lambda \ell^d$, for some $\lambda \in \F^*$ and $\ell \in S_{n,1}(\F)$.
\end{enumerate}
\end{remark}

Observe that, due to properties (1), (2) and (3), the notion  of essential rank induces a metric on $S_{n,d}(\F)$, defined by
\begin{align*} 
\dd_{\ess}(f,g) &\coloneqq \ess(f-g).
\end{align*}
This allows us to develop a theory of error-correcting codes with this metric. 

\begin{definition}
An $[(n,d),k,r]_{\F}$  \textbf{essential-rank-metric code} $\C$ is a $k$-dimensional $\F$-subspace of $S_{n,d}(\F)$ equipped with the  essential-rank metric. The integer $r$ is given by
\[
r:= \min\{ \dd_{\ess}(f,g) \mid f,g \in \C, f\neq g\} =\min \left\{ \ess(f) \mid f \in \C\setminus\{0\} \right\}
\]
and is called the  \textbf{minimum essential-rank distance} of $\C$.
\end{definition}

\begin{example}\label{exa:2}
Let $\F$ be a field with $\mathrm{char}(\F)\neq 2,3$ and consider the $[(3,3),2,2]_\F$ essential-rank-metric code $\C=\langle f_1,f_2\rangle_{\F}$, where $f_1$ is the polynomial in Example \ref{exa:1}, namely 
\[
f_1=3 x_1^3 +  8 x_1^2 x_2 + 5 x_1^2x_3 + 12 x_1 x_2^2 + 4 x_1x_2 x_3 + 4 x_1x_3^2 + 8 x_2^3 + 2 x_2 x_3^2 + x_3^3
\]
and 
\[
f_2=x_1^3+2x_1x_2^2+x_3^3.
\]
It is not easy to determine apriori that the minimum essential-rank distance of $\C$ is $2$. It is clearly $\leq 2$, since $\ess(f_1) \leq 2$. In Section \ref{sec:ess-to-rank}, we will explain how we determine its exact value; see Example \ref{exa:code_2dm}.
\end{example}

The parameters $n,d,k,r$ of a code  are not free, but depend on each other. Specifically, not all the possible choices for the four integers can be achieved by a code. The following result gives a condition that they must satisfy.

\begin{proposition}[Singleton-like bound]\label{prop:singletonlike_bound}
Let $\C$ be an $[(n,d),k,r]_\F$ code. Then 
\[
k \le \binom{n+d-1}{d}-\binom{r+d-2}{d}.
\]
\end{proposition}
\begin{proof}
Consider the space $S_{r-1,d}(\F)$ as a subspace of $S_{n,d}(\F)$. Since every nonzero polynomial in $S_{r-1,d}(\F)$ has essential rank strictly smaller than $r$,  then $\C\cap S_{r-1,d}(\F)=\{0\}$. Thus  
$$
k+\binom{r+d-2}{d}= \dim(\C)+\dim(S_{r-1,d}(\F))\leq \dim(S_{n,d}(\F))=\binom{n+d-1}{d}.
$$
\end{proof}

In Section \ref{sec:ess-to-rank} we will give another bound relating these parameters  with the aid of the theory of rank-metric codes. However, we will also see that it will never be better than  the one of Proposition \ref{prop:singletonlike_bound}.

\subsection{Nondegeneracy and equivalence}

In this section, we introduce the notion of degeneracy and equivalence for essential-rank-metric codes. The first notion captures the effective ambient space where the code should be considered, while the latter notion aims to identify codes that can be considered to be the same.

\begin{definition}
For an $[(n,d),k,r]_\F$ essential-rank-metric code $\C$, we define its space of essential variables, as
\[
V_{\ess}(\C):=\sum_{f\in \C} V_{\ess}(f).
\]
We say that $\C$ is \textbf{nondegenerate} if $V_{\ess}(\C)=S_{n,1}(\F)$, and that it is \textbf{degenerate} if it is not nondegenerate.
\end{definition}

The notion of nondegeneracy allows us to work in the smallest ambient space. Indeed, if an $[(n,d),k,r]_\F$ code $\C$ is degenerate and $\dim(V_{\ess}(\C))=t$, then $\C$  can be isometrically embedded in $S_{t,d}(\F)$.\bigskip

Let us consider natural transformations of the ambient space which keep the metric invariant. In other words, we consider $\F$-linear transformations of $(S_{n,d}(\F),\dd_{\ess})$ which are isometries. For any $A\in \GL(n,\F)$ and $\lambda \in \F^*$, consider induced map
\begin{equation}\label{eq:equiv}
\begin{array}{rcl}
S_{n,d}(\F) & \longrightarrow & S_{n,d}(\F) \\
f(x) & \longmapsto & \lambda f\cdot A:= \lambda f(x\cdot A).
\end{array}
\end{equation}
It is straightforward to verify that this $\F$-linear map is an isometry of $(S_{n,d}(\F),\dd_{\ess})$. We conjecture that, al least under some natural conditions, every isometry is of this form.

\begin{conjecture}
Assume that $\mathrm{char}(\F)=0$ or $\mathrm{char}(\F)>d$. Then every isometry of $(S_{n,d}(\F),\dd_{\ess})$ is of the form $f\mapsto\lambda f\cdot A$ for some $A\in \GL(n,\F)$ and $\lambda \in \F^*$.
\end{conjecture}

In this paper, we will restrict our attention to this known class of isometries.

\begin{definition}
Two $[(n,d),k,r]_\F$ essential-rank-metric codes $\C_1,\C_2$ are said to be \textbf{equivalent} if there exist $A \in\GL_n(\F)$ and $\lambda \in \F^*$ such that
\[
\C_2=\lambda \C_1\cdot A:=\{\lambda f\cdot A \mid f \in \C_1\}.
\]
In this case, we write $\C_1\cong \C_2$.
\end{definition}

\section{Construction of Codes with Large Essential Rank}\label{sec:construction}

In this section we derive a construction of essential-rank-metric codes with large minimum distance. We will do this by exploiting an auxiliary cyclic Galois extension of $\F$ of degree $n$, and generalizing the construction of Delsarte-Gabidulin codes for the rank metric.

Let $\LL$ be a cyclic Galois extension of $\F$ of degree $n$. Let $\sigma$ be a generator of $\Gal(\LL/\F)$. Consider the related algebra extensions 
\[
S_{n}(\LL)= S_n(\F)\otimes_{\F}\LL=\LL[\xxi], \qquad  {S}_{n,d}(\LL)= S_{n,d}(\F)\otimes_{\F}\LL=\LL[\xxi]_d,\]\[{T}_{n}(\LL)= T_n(\F)\otimes_{\F}\LL=\LL[\partial_1,\ldots,\partial_n], \qquad {T}_{n,d}(\LL)= T_{n,d}(\F)\otimes_{\F}\LL=\LL[\partial_1,\ldots,\partial_n]_d.
\]
Let $\alpha^{(1)},\ldots,\alpha^{(d)} \in \LL^n$ be $\F$-bases of $\LL$, and let $1\le \rho \le n$. We define the space 
\[
\C_{\rho}^{n,d}(\alpha^{(1)},\ldots,\alpha^{(d)}) := \left\{ \alpha^{(1)}(\partial)\prod_{j=2}^d\sigma^{r_j}(\alpha^{(j)})(\partial) \,\middle|\, 0 \le r_2,\ldots,r_d \le \rho-2  \right\}^\perp\cap S_{n,d}(\F).
\]
This space is the generalization of symmetric Delsarte-Gabidulin codes when $d>2$. We will see in Section~\ref{sec:ess-to-rank} that for $d=2$ our construction coincides with the construction given in \cite{schmidt2015symmetric}. The remainder of this section is devoted to analysing the parameters of these codes: the minimum essential-rank distance and the dimension.

\subsection{Minimum distance}

This subsection is dedicated to proving the following lower bound on the minimum essential-rank distance. 

\begin{theorem}\label{thm:minessrank}
For every $f\in \C_{\rho}^{n,d}(\alpha^{(1)},\ldots,\alpha^{(d)})\setminus \{0\}$, we have $\ess(f)\geq \rho$.
\end{theorem}

In order to prove Theorem \ref{thm:minessrank}, we need some auxiliary results. We start with the case $d=1$.

\begin{lemma}\label{lem:1d}
We have $\C_{\rho}^{n,1}(\alpha)=\{0\}$.
\end{lemma}
\begin{proof}
Let $f\in\C_{\rho}^{n,1}(\alpha)$ and write $f=c(x)$ for some $c=(c_1,\ldots,c_n) \in \F^n$. Then
\[
0=\alpha(\partial)\circ c(x)=c(\alpha)=\sum_{i=1}^n c_i\alpha_i.
\]
Since $\alpha_1,\ldots,\alpha_n$ are linearly dependent over $\F$, this means that $c_1,\ldots,c_n=0$. So $f=0$.
\end{proof}

The following lemma represents the inductive step for proving Theorem \ref{thm:minessrank}. 

\begin{lemma}\label{lem:d_induction}
Let $d\geq 2$, $f\in\mathcal{C}_\rho^{n,d}(\alpha^{(1)},\ldots,\alpha^{(d)})$ and suppose that $f=g(\ell_1,\ldots,\ell_{\rho-1})$ for some $g\in\F[y_1,\ldots,y_{\rho-1}]_d$ and linearly independent $\ell_1,\ldots,\ell_{\rho-1}\in S_{n,1}(\F)$. Then 
\[
\frac{\partial g}{\partial y_i}(\ell_1,\ldots,\ell_{\rho-1})\in\C_{\rho}^{n,d-1}(\alpha^{(1)},\ldots,\alpha^{(d-1)})
\]
for all $i\in\{1,\ldots,\rho-1\}$.
\end{lemma}
\begin{proof}
Let $r=(r_2,\ldots,r_{d-1})\in \{0,\ldots, \rho-2\}^{d-2}$, $r_{d}\in\{0,\ldots, \rho-2\}$ and write 
\[
W_{d-1}^r(\partial):=\alpha^{(1)}(\partial)\sigma^{r_2}(\alpha^{(2)})(\partial)\cdots\sigma^{r_{d-1}}(\alpha^{(d-1)})(\partial).
\]
We have
\begin{eqnarray*}
0&=&(W_{d-1}^r(\partial)\sigma^{r_{d}}(\alpha^{(d)})(\partial))\circ f\\
&=&(W_{d-1}^r(\partial)\sigma^{r_{d}}(\alpha^{(d)})(\partial))\circ g(\ell_1,\ldots,\ell_{\rho-1})\\
&=&W_{d-1}^r(\partial) \circ (\sigma^{r_{d}}(\alpha^{(d)})(\partial)\circ g(\ell_1,\ldots,\ell_{\rho-1}))\\
&=&W_{d-1}^r(\partial) \circ \sum_{i=1}^{\rho-1}\sigma^{r_d}(A_i)\cdot\frac{\partial g}{\partial y_i}(\ell_1,\ldots,\ell_{\rho-1})\\
&=&\sum_{i=1}^{\rho-1}\sigma^{r_d}(A_i)\cdot\Big( W_{d-1}^r(\partial) \circ \frac{\partial g}{\partial y_i}(\ell_1,\ldots,\ell_{\rho-1})\Big),
\end{eqnarray*}
where $A_i:=\alpha^{(d)}(\partial)\circ \ell_i=\ell_i(\alpha^{(d)})$. So we see that
\begin{equation}\label{eq:linsyst_Moore}
\begin{pmatrix}
A_1&\cdots&A_{\rho-1}\\
\vdots&&\vdots\\
\sigma^{{\rho-2}}(A_1)&\cdots&\sigma^{{\rho-2}}(A_{\rho-1})
\end{pmatrix}\begin{pmatrix}
W_{d-1}^r(\partial) \circ \frac{\partial g}{\partial y_1}(\ell_1,\ldots,\ell_{\rho-1})\\
\vdots\\
W_{d-1}^r(\partial) \circ\frac{\partial g}{\partial y_{\rho-1}}(\ell_1,\ldots,\ell_{\rho-1})
\end{pmatrix}=\begin{pmatrix}
0\\
\vdots\\
0
\end{pmatrix}
\end{equation}
for all $r \in \{0,\ldots, \rho-2\}^{d-2}$. By   \cite[Corollary 4.13]{lam1988vandermonde}, the matrix defining the linear system \eqref{eq:linsyst_Moore} is invertible, and we find that 
\[
W_{d-1}^r(\partial) \circ \frac{\partial g}{\partial y_i}(\ell_1,\ldots,\ell_{\rho-1})=0
\]
for all $r \in \{0,\ldots, \rho-2\}^{d-2}$ and $i\in\{1,\ldots,{\rho-1}\}$. Hence, 
\[
\frac{\partial g}{\partial y_i}(\ell_1,\ldots,\ell_{\rho-1})\in\mathcal{C}_\rho^{n,d-1}(\alpha^{(1)},\ldots,\alpha^{(d-1)})
\]
for all $i\in\{1,\ldots,{\rho-1}\}$.
\end{proof}

We are now ready to give a proof of Theorem \ref{thm:minessrank}.

\begin{proof}[Proof of Theorem \ref{thm:minessrank}]
We will prove now that, if $f\in\mathcal{C}_\rho^{n,d}(\alpha^{(1)},\ldots,\alpha^{(d)})$ is such that 
\[
f=g(\ell_1,\ldots,\ell_{\rho-1})
\]
for some $g\in\F[y_1,\ldots,y_{\rho-1}]_d$ and linearly independent $\ell_1,\ldots,\ell_{\rho-1}\in S_{n,1}(\F)$, then necessarily $f=0$.
We do this by induction on $d$. For $d=1$, we have $f\in\C_{\rho}^{n,1}(\alpha)=\{0\}$ by Lemma \ref{lem:1d}, and so $f=0$. 
For $d\geq 2$, by Lemma \ref{lem:d_induction} we know that 
\[
\frac{\partial g}{\partial y_i}(\ell_1,\ldots,\ell_{\rho-1})\in\mathcal{C}_{\rho}^{n,d-1}(\alpha^{(1)},\ldots,\alpha^{(d-1)})
\]
for all $i\in\{1,\ldots,\rho-1\}$. By our induction assumption, it follows that
\[
\frac{\partial g}{\partial y_i}(\ell_1,\ldots,\ell_{\rho-1})=0
\]
for all $i\in\{1,\ldots,\rho-1\}$. Since $\ell_1,\ldots,\ell_{\rho-1}$ are linearly independent, we get $\frac{\partial g}{\partial y_i}=0$ for all $i\in\{1,\ldots,\rho-1\}$. Hence
\[
g=\frac{1}{d}\sum_{i=1}^{\rho-1}y_i \frac{\partial g}{\partial y_i}=0
\]
and so $f=g(\ell_1,\ldots,\ell_{\rho-1})=0$.
\end{proof}

\begin{example}
Let $(n,d,\rho)=(4,3,3)$. Consider the field $\F=\F_5$ with extension field $\LL=\F_{5^4}=\F_5(\gamma)$, where $\gamma$ is a primitive element of $\LL$ satisfying $\gamma^4=\gamma^2+\gamma+3$. Take $\alpha^{(1)}=\alpha^{(2)}=\alpha^{(3)}=(1,\gamma,\gamma^2,\gamma^3)=:\alpha$. Let 
\[
\mathcal{B}=(x_1^3, x_1^2x_2,\ldots,x_3x_4^2, x_4^3)
\]
be the basis of $\F_5[x_1,x_2,x_3,x_4]_3$ consisting of all monomials ordered lexicographically. Then $\C_3^{4,3}(\alpha,\alpha,\alpha)$ has generator matrix
\setcounter{MaxMatrixCols}{20}
\[
\begin{pmatrix}
1&0&0&0&0&0&0&0&4&4&2&2&3&1&3&4&2&3&3&4\\
0&1&0&0&0&0&0&0&1&2&1&1&3&1&1&1&0&4&2&4\\
0&0&1&0&0&0&0&0&1&0&0&1&0&1&3&0&1&4&4&0\\
0&0&0&1&0&0&0&0&0&1&0&0&1&0&1&3&0&1&4&4\\
0&0&0&0&1&0&0&0&4&4&1&3&4&3&1&1&3&0&4&1\\
0&0&0&0&0&1&0&0&1&3&1&0&1&2&2&4&3&3&3&4\\
0&0&0&0&0&0&1&0&4&3&0&0&0&0&3&3&0&1&3&1\\
0&0&0&0&0&0&0&1&3&3&1&0&1&4&4&4&2&0&2&4
\end{pmatrix}
\]
relative to the basis $\mathcal{B}$. The first row of the matrix corresponds to the polynomial
\[
f:=x_1^3 - x_1x_3x_4 - x_1x_4^2 + 2x_2^3 + 2x_2^2x_3 - 2x_2^2x_4 + x_2x_3^2 - 2x_2x_3x_4 - x_2x_4^2 + 2x_3^3 - 2x_3^2x_4 - 2x_3x_4^2 - x_4^3
\]
with $\ess(f)=4$.
\end{example}

\subsection{Dimension}
In this subsection we give estimates on the dimension of our codes. For simplicity, we restrict to codes of the form
\[
\C_{\rho}^{n,d}(\alpha):=\C_{\rho}^{n,d}(\alpha,\ldots,\alpha).
\]
We first show that the dimension of $\C_{\rho}^{n,d}(\alpha^{(1)},\ldots,\alpha^{(d)})$ does not depend on the choice of the $\F$-bases $\alpha^{(1)},\ldots,\alpha^{(d)}$, and hence that this restriction is harmless.

\begin{proposition}
Let $\alpha^{(1)},\ldots,\alpha^{(d)},\beta$ be $\F$-bases of $\LL$. Then:
\begin{enumerate}
\item We have $\C_{\rho}^{n,d}(\alpha^{(1)},\ldots,\alpha^{(d)})\cong \C_{\rho}^{n,d}(\beta,\alpha^{(2)},\ldots,\alpha^{(d)})$.

\item We have $\C_{\rho}^{n,d}(\alpha^{(1)},\ldots,\alpha^{(d)})= \C_{\rho}^{n,d}(\alpha^{(1)},\alpha^{(\pi(2))},\ldots,\alpha^{(\pi(d))})$ for every permutation $\pi$ of $\{2,\ldots,d\}$.
     
\item We have $\C_{\rho}^{n,d}(\alpha^{(1)},\ldots,\alpha^{(d)})\cong \C_{\rho}^{n,d}(\alpha^{(1)},\ldots,\alpha^{(d-1)},\beta)$.
\end{enumerate}
\end{proposition}
\begin{proof}
(1)
Let $A\in \GL_n(\F)$ be such that $A\alpha^{(1)}=\beta$ and let $f \in \C_{\rho}^{n,d}(\beta,\alpha^{(1)},\ldots,\alpha^{(d)})$. Then, for every $r=(r_2,\ldots,r_d)\in \{0,\ldots \rho-2\}^{d-1}$, we have
\begin{align*}
0&=\beta(\partial) \circ (Z_{d-1}^r(\partial)\circ f) \\
&=\alpha^{(1)}(\partial \cdot  A)\circ (Z_{d-1}^r(\partial)\circ f)\\
&=\alpha^{(1)}(\partial)\circ (Z_{d-1}^r(\partial)\circ f)(x\cdot A) \\
&=\alpha^{(1)}(\partial)Z_{d-1}^r(\partial) \circ (f(x\cdot A)).
\end{align*}
where $Z_{d-1}^r(\partial)=\prod_{j=2}^d\sigma^{r_j}(\alpha^{(j)})(\partial)$. This shows that 
\[
\C_{\rho}^{n,d}(\alpha^{(1)},\ldots,\alpha^{(d)})\cdot A = \C_{\rho}^{n,d}(\beta,\alpha^{(2)},\ldots,\alpha^{(d)}).
\]

\noindent (2) 
This is clear, by commutativity in $T_{n}(\LL)$.

\noindent (3)
The proofs goes as in (1), using the additional property that matrices $A\in \GL(n,\F)$ are invariant under $\sigma$.
\end{proof}

So from now on, we assume that $\alpha:=\alpha^{(1)},\ldots,\alpha^{(d)}$ are all equal. We first give two alternate descriptions of $\C_{\rho}^{n,d}(\alpha)$ in this case.

\begin{lemma}\label{lem:up_bound_codim}
We have
\begin{align*}
\C_{\rho}^{n,d}(\alpha) &= \left\{ \prod_{j=1}^d\sigma^{r_j}(\alpha)(\partial) \,\middle|\, 0 \le r_1,\ldots,r_d \le \rho-2  \right\}^\perp\cap S_{n,d}(\F)\\
&=\left\{ \prod_{j=1}^d\sigma^{r_j}(\alpha)(\partial) \,\middle|\, 0 = r_1\leq r_2\leq \cdots\leq r_d \le \rho-2  \right\}^\perp\cap S_{n,d}(\F).
\end{align*}
\end{lemma}
\begin{proof}
Consider the first equality. Clearly the inclusion $\supseteq$ holds. Let $f\in \C_{\rho}^{n,d}(\alpha)$. To show that equality holds, we need to show that $\prod_{j=1}^d\sigma^{r_j}(\alpha)(\partial)\circ f=0$ for all 0 $\le r_1,\ldots,r_d \le \rho-2$.
As the $\sigma^{r_j}(\alpha)(\partial)$'s commute with each other, we already know that this holds when $r:=\min(r_1,\ldots,r_d)=0$. In general, we note that
\[
\prod_{j=1}^d\sigma^{r_j}(\alpha)(\partial)\circ f=\sigma^r\left(\prod_{j=1}^d\sigma^{r_j-r}(\alpha)(\partial)\circ f\right)=\sigma^r(0)=0
\]
as $r_j-r=0$ for some $j\in\{1,\ldots,d\}$. So the first equality holds. The second equality now easily follows from the fact that the $\sigma^{r_j}(\alpha)(\partial)$'s commute with each other.
\end{proof}

Lemma \ref{lem:up_bound_codim} gives us our first upper bound on the codimension of $\C_{\rho}^{n,d}(\alpha)$. Namely, the subspace $\C_{\rho}^{n,d}(\alpha)$ is defined by
\begin{align*}
\binom{(d-1)+(\rho-2)}{d-1}&=\#\{(\delta_1,\ldots,\delta_d)\in\Z_{\geq0}^d\mid \delta_1+\ldots+\delta_d=\rho-2\}
\\
&=\#\{(r_2,\ldots,r_d)\in\Z_{\geq0}^{d-1}\mid 0=r_1\leq r_2\leq\cdots \leq r_d\leq r_{d+1}=\rho-2\}
\end{align*}
equations over $\LL$, each of which can be interpreted as $n$ equations over $\F$. Here we use $\delta_i=r_{i+1}-r_i$ for the second equality. We conclude the following result.

\begin{proposition}\label{prop:dimbound}
We have
\[
\dim \C_{\rho}^{n,d}(\alpha)\geq\binom{n+d-1}{d} -n\binom{d+\rho-3}{d-1}.
\]
\end{proposition}

For a more precise lower bound on the dimension of $\C_\rho^{n,d}(\alpha)$,  we have to better understand the conditions
\[
\prod_{j=1}^d\sigma^{r_j}(\alpha)(\partial)\circ f=0
\]
and how they relate to each other.

\begin{lemma}\label{lm:same_condition}
Let $f\in\F[x_1,\ldots,x_n]_d$ and $r_1,\ldots,r_d,s_1,\ldots,s_d\in\Z$. Suppose that there exists an $x\in\Z$ such that $r_j=x+s_j\mod n$ for all $j\in\{1,\ldots,d\}$. Then
\[
\prod_{j=1}^d\sigma^{r_j}(\alpha)(\partial)\circ f=0\mbox{\quad if and only if \quad}\prod_{j=1}^d\sigma^{s_j}(\alpha)(\partial)\circ f=0.
\]
\end{lemma}
\begin{proof}
We have
\[
\prod_{j=1}^d\sigma^{r_j}(\alpha)(\partial)\circ f=\prod_{j=1}^d\sigma^{x+s_j}(\alpha)(\partial)\circ f=\sigma^x\left(\prod_{j=1}^d\sigma^{s_j}(\alpha)(\partial)\circ f\right)
\]
as $\sigma^n=\mathrm{id}$ and so the stated equivalence follows.
\end{proof}

The following example shows that, even when restricting to the case $r_1=0$ and $r_2\leq\cdots\leq r_d$, the conditions defining $\C_{\rho}^{n,d}(\alpha)$ can be redundant.

\begin{example}
Let $d=3$, $f\in\F[x_1,\ldots,x_n]_d$, write $x:=\rho-2$ and suppose that $n=2x$ holds. Then
\[
\alpha(\partial)\alpha(\partial)\sigma^x(\alpha)(\partial)\circ f=0\mbox{\quad if and only if \quad}\alpha(\partial)\sigma^x(\alpha)(\partial)\sigma^x(\alpha)(\partial)\circ f=0
\]
as $(x,x,0)\equiv (0,0,x)+(x,x,x)\mod n$ can be reordered to $(0,x,x)$.
\end{example}

\begin{remark}
When $0<\mathrm{char}\,\F\leq d$, there is another kind of redundancy: we have
\[
\alpha(\partial)^d\circ f=d!f(\alpha)=0
\]
and so the condition $\alpha(\partial)^d\circ f=0$ is trivially fulfilled.
\end{remark}

As the $r_j$'s are exponents of $\sigma$, we can naturally view them as elements of $\Z/n\Z$. And as the order of the $r_j$'s does not matter, we see that the conditions defining $\C_{\rho}^{n,d}(\alpha)$ can be indexed by multisets $S=\{\!\!\{r_1,\ldots,r_k\}\!\!\}$. For such a multiset $S$ and an integer $x\in\Z$, we define
\[
x+S=\{\!\!\{x+y\mod n\mid y\in S\}\!\!\}.
\]
This operation defines an action of $\Z$ on the set $\mathcal{M}^{n,d}$ of $d$-element multisets whose elements all lie in $\Z/n\Z$. Lemma~\ref{lm:same_condition} states that two multisets in the same $\Z$-orbit define equivalent conditions. Hence, the codimension of $\C_{\rho}^{n,d}(\alpha)$ is at most $n$ times the number of orbits that contain a multiset of the form $\{\!\!\{r_1,\ldots,r_d\}\!\!\}$ with $0\leq r_1,\ldots,r_d\leq \rho-2$. We conjecture that equality holds.

\begin{conjecture}
Suppose that $\mathrm{char}\,\F=0$ or $\mathrm{char}\,\F>d$ and that $1\leq \rho\leq n$. Then the codimension of $\C_{\rho}^{n,d}(\alpha)$ in $\F[x_1,\ldots,x_n]_d$ equals $n$ times the number of $\Z$-orbits in $\mathcal{M}^{n,d}$ that contain a multiset of the form $\{\!\!\{r_1,\ldots,r_d\}\!\!\}$ with $0\leq r_1,\ldots,r_d\leq \rho-2$.
\end{conjecture}

See the appendix for more on the number of such orbits.\medskip
 
We conclude this section by investigating how far our codes are from being optimal with respect to the Singleton-like bound of Proposition \ref{prop:singletonlike_bound}.
By Propositions~\ref{prop:singletonlike_bound} and~\ref{prop:dimbound}, we have
\[
s(n,d,\rho)\leq\dim \C_{\rho}^{n,d}(\alpha)\leq r(n,d,\rho).
\]
for
\begin{align*}
s(n,d,\rho)&:=\binom{n+d-1}{d} -n\binom{\rho+d-3}{d-1},\\
r(n,d,\rho)&:=\binom{n+d-1}{d}-\binom{\rho+d-2}{d}.
\end{align*}
We now compare the values of $s(n,d,\rho)$ and $r(n,d,\rho)$ asymptotically in order to determine how close our codes are from being optimal.

First, when $d,\rho$ are fixed, we have
\[
\begin{array}{ccr}
s(n,d,\rho)&\sim&n^d/d!\\
r(n,d,\rho)&\sim&n^d/d!
\end{array}
\]
as $n$ grows large. We see that our codes have the maximal dimension, according to the Singleton-like bound. When $d$ is fixed and $\rho=\tau n$ for some constant $0<\tau<1$, then
\[
\begin{array}{cccc}
s(n,d,\rho)&\sim&(1-d\tau^{d-1})&\cdot\, n^d/d!\\
r(n,d,\rho)&\sim&(1-\tau^d)&\cdot\, n^d/d!
\end{array}
\]
In this case, our codes are a factor $\frac{1-\tau^d}{1-d\tau^{d-1}}$  away from attaining the upper bound of Proposition \ref{prop:singletonlike_bound}. However, this might be due to several facts. The first one is that the lower bound on the dimension of $\C_{\rho}^{n,d}(n,d)$ can be improved by further investigation on the $\mathbb Z$-orbits of $\mathcal M^{n,d}$, as explained above. On the other hand, we  believe that also the bound of Proposition \ref{prop:singletonlike_bound} can be improved. Indeed, we will see later in Section \ref{sec:ess-to-rank} that for $d=2$ and when $\F$ is finite, there is a better upper bound proved by Schmidt \cite{schmidt2015symmetric}.

\section{A decoding algorithm}\label{sec:decoding}

In this section, for the sake of simplicity we again choose $\alpha:=\alpha^{(1)},\ldots,\alpha^{(d)}$ to be equal and write $\C_\rho^d(\alpha):=\C_{\rho}^{n,d}(\alpha,\ldots,\alpha)$.
Suppose that $f\in \C_\rho^d(\alpha)$  is sent and $F=f+g'$  is received,  such that  $\ess(g')=e\leq (\rho-1)/2$. Thus, we have $g'=g(\ell_1,\ldots,\ell_e)$ for some $g\in S_{e,d}(\F)$ and linearly independent $\ell_1,\ldots,\ell_e\in S_{n,1}(\F)$. We aim to reconstruct $f$. To this end, we divide the decoding in two parts: we first   reconstruct $V_{\ess}(g')=\langle \ell_1,\ldots,\ell_e\rangle_{\F}$, and then we recover~$g'$. Knowing $F$ and $g'$ together, we can clearly reconstruct the original codeword $f$. 

\subsection{Recover \texorpdfstring{$V_{\ess}(g')$}{V\_ess(g')}} For a polynomial $h\in S_n(\LL)$, we define $\coeff_1(h),\ldots,\coeff_n(h)\in S_n(\F)$ to be the polynomials such that
\[
h=\alpha_1\coeff_1(h)+\ldots+\alpha_n\coeff_n(h).
\]
For $k\in\{1,\ldots,d\}$ and  $r=(r_1,\ldots,r_{k}) \in \{0,\ldots,\rho-2\}^{k}$, write 
\[
Z_{k}^{r}(\partial):=\prod_{j=1}^{k}\sigma^{r_j}(\alpha)(\partial) \in T_{n,k}(\LL).
\]
We have the following result.

\begin{proposition}\label{prop:vess}
The space $V_{\ess}(g')$ is spanned by 
\[
\left\{\coeff_i\left(Z_{d-1}^r(\partial) \circ g(\ell_1,\ldots,\ell_e)\right) \,\middle|\,i\in\{1,\ldots,n\}, r\in\{0,\ldots,\rho-2\}^{d-1}, r_1=0\right\}.
\]
\end{proposition}
\begin{proof}
By changing coordinates, we may assume that $\ell_j=x_j$ for each $j \in \{1,\ldots,e\}$.  Now, for $k\geq1$ define $F_k\subseteq T_{e,k}(\F)$ to be the subspace spanned by all elements
\[
\left\{\coeff_i\left(Z_{k}^r(\partial)|_{\partial_{e+1},\ldots,\partial_n=0} \right) \,\middle|\,i\in\{1,\ldots,n\}, r\in\{0,\ldots,\rho-2\}^{k}, r_1=0\right\}.
\]
We claim that $F_k=T_{e,k}(\F)$.\medskip

Recall that the map $\Phi_{e,k}\colon S_{e,k}(\F) \times T_{e,k}(\F)\longrightarrow \F$ defined in \eqref{eq:pairing} is a perfect pairing. If $F_k\neq T_{e,k}(\F)$, then $F_k^\perp$ contains a nonzero polynomial $h\in S_{e,k}(\F)$. View $h$ as an element of $S_{n,k}(\F)$. Then $h\in\mathcal{C}_\rho^k(\alpha)$. Indeed, as
\[
\frac{\partial h}{\partial x_{e+1}}=\cdots=\frac{\partial h}{\partial x_n}=0
\]
we see that
\begin{align*}
0&=\Phi_{e,k}\left(h,\coeff_i\left(Z_{k}^r(\partial)|_{\partial_{e+1},\ldots,\partial_n=0} \right)\right)\\
&=\coeff_i\left(Z_{k}^r(\partial)|_{\partial_{e+1},\ldots,\partial_n=0} \right)\circ h \\
&=\coeff_i\left((Z_{k}^r(\partial)|_{\partial_{e+1},\ldots,\partial_n=0})\circ h\right)\\
&=\coeff_i\left(Z_{k}^r(\partial)\circ h\right)
\end{align*}
for all $i\in\{1,\ldots,n\}$ and $r\in\{0,\ldots,\rho-2\}^{k}$. However,  $\ess(h)\leq e\leq (\rho-1)/2$ and every nonzero element of $\C_\rho^k(\alpha)$ has essential rank $\geq \rho$. This is a contradiction, hence $F_k=T_{e,k}(\F)$.\bigskip

Finally, we note that
\[
V_{\ess}(g')=\langle x_1,\ldots,x_e\rangle_{\F}=\Phi_{e,d,d-1}(g(x_1,\ldots,x_e),T_{e,d-1}(\F))=\Phi_{e,d,d-1}(g(x_1,\ldots,x_e),F_{d-1})
\]
is exactly spanned by the stated set.
\end{proof}

We can actually improve Proposition \ref{prop:vess}, observing that we can reduce the number of generators for the space $V_{\ess}(g')$. Indeed, since permutations of the entries of $r$ yield the same $Z_k^r(\partial)$, we can actually restrict to consider those $r$ whose entries are nondecreasing. In other words, let us define
\[
X_{k,\rho}:=\{ (0,r_2,\ldots,r_k) \in \{0,\ldots,\rho-2\}^k \mid 0\leq r_2 \leq \ldots\leq r_k \}.
\]

\begin{corollary}\label{cor:vess}
The space $V_{\ess}(g')$ is spanned by 
\[
\left\{\coeff_i\left(Z_{d-1}^r(\partial) \circ g(\ell_1,\ldots,\ell_e)\right) \,\middle|  \, r\in X_{d-1,\rho} \right\}.
\]
\end{corollary}

Thanks to Corollary \ref{cor:vess}  we know that, in order to reconstruct $V_{\ess}(g')$, it suffices to efficiently recover
\[
Z_{d-1}^r (\partial) \circ g(\ell_1,\ldots,\ell_e)=:h^{r}(x)
\]
for $r\in X_{d-1,\rho}$. Write $h^r_i(x):=\coeff_i(h^r)$, so that
\[
h^r(x)=\alpha_1 \coeff_1(h^{r})+\ldots+\alpha_n \coeff_n(h^{r})= \alpha_1 h_1^{r}(x)+\ldots+\alpha_n h_n^{r}(x).
\]
The problem of computing $(h^r_1(x),\ldots,h_n^r(x))$ can be transformed into an instance of decoding a certain Gabidulin code in $\LL^n\cong \F^{n\times n}$ with minimum distance $\rho$. We explain this connection in the following.

\subsubsection{Decoding syndromes}
For every $u\in X_{d,\rho}$, we compute the \textbf{$u$-syndrome} of $F$, which is defined as  
\[
s_{u}:=Z_d^u(\partial)\circ F \in \LL.
\]
We write $s_v=s_u$ when $v\in\{0,\ldots,\rho-2\}^d$ can be reordered to $u\in X_{d,\rho}$. Take $\beta:=\sigma^{-\rho+2}(\alpha)\in\LL^n$ and define $\gamma=(\gamma_1,\ldots,\gamma_n)\in\LL^n$ to be the unique solution (up to $\LL$-scalar multiple) of 
\begin{equation}\label{eq:gamma}
\sum_{j=1}^n\sigma^{i}(\beta_j)\gamma_j=0, \qquad \mbox{ for all } i\in\{0,\ldots,n-1\}\setminus \{n-\rho+1\}.
\end{equation}

We consider the following Gabidulin code
\begin{align} \label{eq:Gabidulin}
\mathcal{G}_{\rho,n}^{\sigma}(\gamma):=&\{(\mathrm{ev}_{\gamma_1}(f),\ldots,\mathrm{ev}_{\gamma_n}(f)) \mid f\in \LL[\sigma], \deg_{\sigma}(f)\le n-\rho\} \subseteq \LL^n,
\end{align}
  where, for $f=\sum_if_i\sigma^i \in \LL[\sigma]$ and $\alpha \in \LL$, we define $\mathrm{ev}_{\alpha}(f):=\sum_if_i\sigma^i(\alpha)$. 

Gabidulin codes represent the most prominent family of codes in the rank metric, due to their mathematical features and efficient decoding algorithms. They were first introduced by Delsarte  \cite{delsarte1978bilinearforms}, and later rediscovered by Gabidulin \cite{gabidulin1985maxrankdistance}. Among their properties, the most remarkable one is that they have the maximum possible minimum rank distance for the given parameters. More specifically, we can endow the space $\LL^n$ with a weight function, namely the \textbf{$\F$-rank weight}, given by 
$$\rk_\F(v)=\dim_\F\langle v_1,\ldots,v_n\rangle_\F, \qquad \mbox{ for any } v=(v_1,\ldots,v_n) \in \LL^n.$$
The rank weight induces the so-called  \textbf{rank distance} on $\LL^n$. Such a metric space is $\F$-linearly isometric to the metric space $(\F^{n\times n},\dd_{\rk})$; see Section \ref{sec:ess-to-rank}.
It was shown in \cite{delsarte1978bilinearforms,gabidulin1985maxrankdistance} that, for any $1 \le \rho \le n$, the minimum rank weight of a nonzero codeword in a Gabidulin code $\mathcal G_{\rho,n}^\sigma(\gamma)$ is $\rho$. Moreover, a parity-check matrix for $\mathcal G_{\rho,n}^\sigma(\gamma)$ in \eqref{eq:Gabidulin} is 
\begin{equation}\label{eq:PC_Gab}
   H_{\rho,n}^{\sigma}(\beta)= \begin{pmatrix}
        \beta_1 & \beta_2 & \cdots & \beta_n \\
        \sigma(\beta_1) & \sigma(\beta_2) & \cdots & \sigma(\beta_n)\\
        \vdots & \vdots & & \vdots \\
        \sigma^{\rho-2}(\beta_1) & \sigma^{\rho-2}(\beta_2) & \cdots & \sigma^{\rho-2}(\beta_n)
    \end{pmatrix}.
\end{equation}
With this in mind, for any $v\in \LL^n$, we call the vector $v(H_{\rho,n}^{\sigma}(\beta))^\top$ its \textbf{$H$-syndrome}. It is important to mention that there are several  decoding algorithms of Gabidulin codes which make use of the $H$-syndrome calculations. For a given matrix $H_{\rho,n}^{\sigma}(\beta)$ and a  vector $s\in \LL^{\rho-1}$, these algorithms return the unique vector $v\in\LL^n$ of rank weight at most $\frac{\rho-1}{2}$ such that $v(H_{\rho,n}^{\sigma}(\beta))^\top=s$, if any. For more details on decoding of Gabidulin codes we refer the interested reader to \cite[Section 2.8]{bartz2022rank} and references therein.

\begin{proposition}
Let $r_2,\ldots, r_{d-1} \in \{0,\ldots, \rho-2\}$ and take $r=(0,r_2,\ldots,r_{d-1})$. Then 
\[
(h_1^{r}(\alpha),\ldots,h_n^{r}(\alpha))
\]
is the error vector of minimum rank weight  whose $H$-syndrome is $(\sigma^{t-\rho+2}(s_{(r,\rho-2-t)}))_{t=0}^{\rho-2}$ in $\mathcal G_{\rho,n}^\sigma(\gamma)$. 
\end{proposition}
\begin{proof}
Note that
\[
s_u=Z_d^u(\partial)\circ F= Z_d^r(\partial)\circ (f+g')=Z_d^u(\partial)\circ f +Z_d^u(\partial)\circ  g'=Z_d^u(\partial)\circ  g'.
\]
We have 
\[
s_{(r,r_d)}=\sigma^{r_d}(\alpha)(\partial)\circ h^r =\sigma^{r_d}(h_1^{r}(\alpha))\alpha_1+\ldots+\sigma^{r_d}(h_n^{r}(\alpha))\alpha_n,\qquad \mbox{ for all } r_d\in \{0,\ldots,\rho-2\},
\]
or equivalently, by applying $\sigma^{-r_d}$ to every equation,
$$ \sigma^{-r_d}(s_{(r,r_d)})=h_1^{r}(\alpha)\sigma^{-r_d}(\alpha_1)+\ldots+h_n^{r}(\alpha)\sigma^{-r_d}(\alpha_n), \qquad \mbox{ for all } r_d\in \{0,\ldots,\rho-2\}.$$ 
Setting $r_d=\rho-2-t$, the latter system of equations can be also rewritten as
\begin{equation}\label{eq:key}
\sigma^{t-\rho+2}(s_{(r,\rho-2-t)})=h_1^{r}(\alpha)\sigma^t(\beta_1)+\ldots+h_n^{r}(\alpha)\sigma^t(\beta_n), \qquad \mbox{ for all } t\in \{0,\ldots,\rho-2\}.
\end{equation}
Since by Proposition \ref{prop:vess} we have $\langle h_1^{r},\ldots,h_n^{r}\rangle_\F\subseteq V_{\ess}(g')$, we can deduce that 
\[
\dim_{\F}\langle h_1^{r}(\alpha),\ldots,h_n^{r}(\alpha)\rangle_\F =\dim_{\F}\langle h_1^{r},\ldots,h_n^{r}\rangle_\F  \leq \dim_{\F}V_{\ess}(g')=e\leq \frac{\rho-1}{2}.
\]
So we see that $(h_1^{r}(\alpha),\ldots,h_n^{r}(\alpha))$ is indeed the vector of minimum rank weight whose $H$-syndrome is $(\sigma^{t-\rho+2}(s_{(r,\rho-2-t)}))_{t=0}^{\rho-2}$.
\end{proof}

Thus, with any decoding algorithm $\mathrm{Dec}(H_{\rho,n}^\sigma(\beta),\cdot)$ for the code $\mathcal G_{\rho,n}^\sigma(\gamma)$, we can first recover the vector $(h_1^{r}(\alpha),\ldots,h_n^{r}(\alpha))=\mathrm{Dec}(H_{\rho,n}^\sigma(\beta),s)$, where $s:=(\sigma^{t-\rho+2}(s_{(r,\rho-2-t)}))_{t=0}^{\rho-2}$. From that, we can then retrieve the vector $(h_1^{r}(x),\ldots,h_n^{r}(x))$ and finally 
\[
h^r=Z_{d-1}^r (\partial) \circ g(\ell_1,\ldots,\ell_e),
\]
as originally planned. 

\subsection{Recover \texorpdfstring{$g'$}{g'}} 
Now that we have found the space $V_{\ess}(g')$, we first choose any $\F$-basis for it, namely $(\ell_1,\ldots,\ell_e)$, and then we do a coordinate change so that $\ell_i=x_i$ for all $i\in\{1,\ldots,e\}$. We wish to recover $g'=g(x_1,\ldots,x_e)\in S_{e,d}(\F)$ given
\begin{equation}\label{eq:system}
s_r=Z_d^r(\partial)\circ g,\quad r\in X_{d,\rho}.
\end{equation}

\begin{lemma}\label{lem:unique}
There is at most one solution to the system of equations
\[
Z_d^r(\partial)\circ g=s_r,\quad r\in X_{d,\rho}.
\]
\end{lemma}
\begin{proof}
Suppose that $g_1,g_2$ are both solutions of these equations. Then $g_1-g_2\in\mathcal{C}_\rho^d(\alpha)$. As $\mathcal{C}_\rho^d(\alpha)$ has minimum distance $\rho$ and $g_1-g_2$ has essential rank $\leq \rho-1$, this means that $g_1-g_2=0$. So the system has at most one solution.
\end{proof}

By Lemma \ref{lem:unique}, we can now indeed recover $g$. We summarize the decoding algorithm below.

\RestyleAlgo{boxruled}
\begin{algorithm}[ht]
 \KwData{$n, d, \rho\in\mathbb{N}$ with $d\geq 2$ and $2\leq \rho \leq n$, an $\F$-basis $\alpha\in\LL^n$ and a received word $F\in S_{n,d}(\F)$ such that there exists an $f\in\C_{\rho}^{n,d}(\alpha)$ with $\ess(F-f)\leq \frac{\rho-1}{2}$}
 \KwResult{The original codeword $f\in \C_{\rho}^{n,d}(\alpha)$}
 $L:=\emptyset$\;
 Compute $\sigma^i(\alpha)$ for $i\in\{1,\ldots,n-1\}$\;
 $\beta:=\sigma^{-\rho+2}(\alpha)$\;
 \For{$u\in X_{d,\rho}$}{
 	$s_u:=Z_{d}^u(\partial)\circ F$\;
 } 
 \For{$r\in X_{d-1,\rho}$}{ 
    $s:=(\sigma^{2-\rho}(s_{(r,\rho-2)}),\sigma^{3-\rho}(s_{(r,\rho-3)}),\ldots,s_{(r,0)})$\;
    $(h_1^r(\alpha),\ldots,h_n^r(\alpha)):=\mathrm{Dec}(H_{\rho,n}^\sigma(\beta),s)$\;
    $L\gets L\cup\{h_1^r(x),\ldots,h_n^r(x)\}$\;
 }

 $(\ell_1,\ldots,\ell_e):=\mathrm{Basis}(\mathrm{span}_{\F}(L),\F)$\;
$g:=\mathrm{Solve}(\eqref{eq:system})$\;
$g':=g(\ell_1,\cdots,\ell_e)$\;
\Return{$F-g'$}
  \caption{Decoding Algorithm for Codes $\C_{\rho}^{n,d}(\alpha)$}\label{alg1}
\end{algorithm}

\subsection{Complexity of the Algorithm \ref{alg1}}

We give our step-by-step analysis of the computational cost of Algorithm \ref{alg1} by expressing all computations in terms of operations over $\F$ and over $\LL$. This is because, depending on the choice of the representations of elements in $\LL$, some operations might be faster or slower. Also, there are faster algorithms for operations over extension fields which are specific over finite fields, For the rest of the section, let $\omega$ be the exponent of matrix multiplication, which is proved to satisfy $\omega\leq 2.37188$ \cite{duan2022faster}.

The first step is computing $\sigma^i(\alpha)$ for $i \in \{1,\ldots,n-1\}$, which can be done in constant time if elements of $\LL$ are represented with respect to a normal basis ($\sigma^i(\alpha)$ is just the $i$th shift of the vector representing $\alpha$), or in $\mathcal O(n^2)$ operations over $\F$, if we represent elements in $\LL$ with respect to any $\F$-basis $\beta_1,\ldots,\beta_n$ of $\LL$ for which the values $\sigma^i(\beta_j)$ are already given. In both cases, this will not be the most expensive step of the algorithm.

We then need to calculate the {$u$-syndrome} of $F$ for every $u\in X_{d,\rho}$. In fact, we compute $Z_{k}^u(\partial)\circ F$ for every $k\in\{1,\ldots,d\}$ and $u\in X_{k,\rho}$ so that we can make use of the fact that
\[
Z_{k+1}^u(\partial)\circ F =\sigma^{u_{k+1}}(\alpha)(\partial)\circ (Z_{k}^{(0,u_2,\ldots,u_k)}(\partial)\circ F)
\]
for every $k\in\{1,\ldots,d-1\}$ and $u\in X_{k+1,\rho}$. So for every $k\in\{1,\ldots,d\}$ and $u\in X_{k,\rho}$, we need to compute $L(\partial)\circ G$ for some $L\in \LL^n$ and $G\in S_{n,d+1-k}(\LL)$, which takes $n\binom{n-1+d-k}{d-k}$ multiplications and $(n-1)\binom{n-1+d-k}{d-k}$ additions over $\LL$. This gives a total of
\[
\sum_{k=1}^d\binom{\rho+k-3}{k-1}(2n-1)\binom{n-1+d-k}{d-k}=\mathcal{O}\left(n^d\right)
\]
operations over $\LL$ is this step.

Now, for every $r \in  X_{d-1,\rho}$ we compute the vector $s$, which has negligible cost if the elements are represented with respect to a normal basis, or $\mathcal O(\rho n)$ operations over $\F$ otherwise. Moreover, we solve the decoding instance $\mathrm{Dec}(H_{\rho,n}^\sigma(\beta),s)$. This can be done in $\mathcal O(n^2)$ operations over $\LL$ when elements are represented with respect to a normal basis; see for instance \cite[Section 2.8]{bartz2022rank}. The total cost of this step, assuming to have chosen a normal basis, is $\mathcal O(n^2|X_{d-1,\rho}|)=\mathcal O(n^2\binom{d+\rho-4}{d-2})$ operations over $\LL$.

After that, we need to extract an $\F$-basis of $\mathrm{span}_{\F}(L)$ from the set $L$, which has size $\leq n|X_{d-1,\rho}|$. Since the elements $h_i^j(x) \in L$ are elements of $S_{n,1}(\F)$, they are vectors of length $n$. The basis computation thus can be done via Gaussian elimination, and it takes $\mathcal O(n^\omega|X_{d-1,\rho}|)=\mathcal O(n^\omega\binom{d+\rho-4}{d-2})$ operations over $\F$; see \cite{bunch1974}.

Finally, we need to solve \eqref{eq:system}. Writing
\[
g=\sum_{t_1+\ldots+t_e=d}a_{t_1,\ldots,t_e}x_1^{t_1}\cdots x_e^{t_e},
\]
we see that \eqref{eq:system} becomes a linear system of $|X_{d,\rho}|$ equations over $\LL$ in the $\binom{e+d-1}{d}$ variables $a_{t_1,\ldots,t_e}$. Since we want a solution over $\F$, we can transform every equation into $n$ equations over $\F$ by choosing any $\F$-basis of $\LL$. This provides a total of 
\[
n|X_{d,\rho}|=n\binom{d+\rho-3}{d-1}
\]
equations over $\F$. Such a linear system can be solved with 
Gaussian elimination, yielding a computational cost of  $\mathcal O( n|X_{d,\rho}|\binom{e+d-1}{d}^{\omega-1})=\mathcal O( n\binom{d+\rho-3}{d-1}\binom{e+d-1}{d}^{\omega-1})$ operations over $\F$ \cite{bunch1974}.\medskip
 
Multiplication over $\LL$ cost $\mathcal O(n^2)$ operations over $\F$, when we represent elements of $\LL$ with respect to a nonspecific basis, while addition costs $\mathcal O(n)$ operations. If elements of $\LL$ are represented in a more clever way, one can speed up multiplication up to reaching $\mathcal O(n)$ operations over $\F$. However, it is not always clear if this is compatible with having a normal basis. Over finite fields, one can have $\F$-bases of $\LL$ which are normal and have linear complexity \cite{mullin}. When we can do this, the total computational cost of our algorithm in terms of operations over $\F$ is 
\[
\mathcal O\left(n^3\binom{d+\rho-4}{d-2}+n\binom{d+\rho-3}{d-1}\binom{e+d-1}{d}^{\omega-1} +n^{d+1}\right).
\]
Assuming $e\simeq \frac{\rho}{2}$ and $\rho=\tau n$ for some $0<\tau<1$, this gives a total cost of
\[
\mathcal O(n^{d\omega})
\]
operations over $\F$. 

\begin{remark}If we restrict to the case $d=2$, we see that the complexity of our algorithm is not optimal. One could simply use a decoder for a classical Gabidulin code containing $\C_\rho^{n,2}(\alpha)$ and having the same minimum rank distance (this is due to the structure of Gabidulin codes), and decode $\C_\rho^{n,2}(\alpha)$ with $\mathcal O(n^2)$ operations over $\LL$ and hence $\mathcal O(n^3)$ operations over $\F$, while we need a total cost of $\mathcal O(n^{2\omega})$ operations over $\F$ using Algorithm \ref{alg1}. This suggests that there must be room for improvement on the speed of a decoder. 
\end{remark}

\section{Essential-rank-metric codes as rank-metric codes}\label{sec:ess-to-rank}

The aim of this section is to bridge the theory of essential-rank-metric codes with the well-established world of codes endowed with the rank distance. We will do this with the aid of a result of Carlini \cite{carlini2006reducing}, which shows how to transform the  essential rank of a homogeneous polynomial $f$ into the rank of a special matrix associated with $f$, namely its \emph{first catalecticant matrix}. In addition, we will also relate the space of essential variables with the \emph{column support} of a matrix, that is the linear space generated by its columns.

\begin{definition}
Let $z_1,\ldots,z_m$ be the monomial basis of $S_{n,d-1}(\F)$ ordered lexicographically, where $m=\binom{n+d-2}{d-1}$. The \textbf{first catalecticant matrix} of $f\in S_{n,d}(\F)$ is the matrix $C_f\in \F^{n \times m}$, where
\[
\partial_i \circ f = \sum_{j=1}^m(C_f)_{i,j}z_j.
\]
\end{definition}

The following theorem is the main result proved by Carlini in \cite{carlini2006reducing}. It relates the essential rank of a homogeneous polynomial with its first catalecticant matrix, and characterizes the space of essential variables.

\begin{theorem}[{\cite[Proposition 10]{carlini2006reducing}}]\label{thm:carlini_ess}
Let $f\in S_{n,d}(\F)$ and assume that $\mathrm{char}(\F)=0$ or $\mathrm{char}(\F)>d$. Then the following hold:
\begin{enumerate}
\item $\ess(f)=\rk(C_f)$
\item $V_{\ess}(f)=\{D\circ f \mid D\in T_{n,d-1}(\F)\}$
\end{enumerate}
\end{theorem}

We illustrate Theorem \ref{thm:carlini_ess} elaborating on Example \ref{exa:1}. 
\begin{example}\label{exa:3} Let
\[
f=3 x_1^3 +  8 x_1^2 x_2 + 5 x_1^2x_3 + 12 x_1 x_2^2 + 4 x_1x_2 x_3 + 4 x_1x_3^2 + 8 x_2^3 + 2 x_2 x_3^2 + x_3^3 \in S_{3,3}(\F)
\]
be as in Example \ref{exa:1} and let us compute the first and second order partial derivatives of $f$: we get
\[
\begin{array}{rl}
\partial_1\circ f&=~9x_1^2+16x_1x_2+10x_1x_3+12x_2^2+4x_2x_3+4x_3^2,\\
\partial_2\circ f&=~8x_1^2+24x_1x_2+4x_1x_3+24x_2^2+2x_3^2,\\
\partial_3\circ f&=~5x_1^2+4x_1x_2+8x_1x_3+4x_2x_3+3x_3^2
\end{array}
\]
and
\[
\begin{array}{rl}
\partial_1^2\circ f&=~18 x_1 + 16 x_2 + 10 x_3,\\
\partial_1\partial_2 \circ f&=~ 16 x_1 + 24 x_2 + 4 x_3,\\
\partial_2^2 \circ f&=~ 24 x_1 + 48 x_2, \\
\partial_1\partial_3 \circ f&=~ 10x_1+4x_2+8x_3,\\
\partial_2\partial_3 \circ f &=~ 4 x_1 + 4 x_3, \\
\partial_3^2 \circ f&=~ 8x_1+4x_2+6x_3. 
\end{array}
\]
So the first catalecticant matrix of $f$ is
\[
C_f=\begin{pmatrix}
9 & 16 & 10 & 12 & 4 & 4 \\
8 & 24 & 4 &  24 & 0 & 2 \\
5 &  4 & 8 &  0 & 4 & 3
\end{pmatrix}.
\]
By Theorem \ref{thm:carlini_ess}, we get
\[
\ess(f)=\rk(C_f)=2
\]
and
\begin{align*}
V_{\ess}(f)&=\langle \partial_1^2\circ f, \partial_1\partial_2 \circ f, \partial_2^2\circ f, \partial_1\partial_3 \circ f, \partial_2\partial_3 \circ f, \partial_3^2\circ f \rangle_{\F}\\
&=\langle x_1+2x_2,x_1+x_3 \rangle_{\F}.
\end{align*}
\end{example}

As a consequence of Theorem \ref{thm:carlini_ess}, we can embed the metric space $(S_{n,d}(\F),\dd_{\ess} )$ isometrically into the space of $n \times \binom{n+d-2}{d-1}$ matrices over $\F$ equipped with the distance function $\dd_{\rk}$ defined as 
\[
\dd_{\rk}(A,B)=\rk(A-B),\quad A,B \in \F^{n \times \binom{n+d-2}{d-1}},
\]
which is called the \textbf{rank distance}.

\begin{corollary}\label{cor:embedding}
Assume that $\mathrm{char}(\F)=0$ or $\mathrm{char}(\F)>d$. The map     
\[
\begin{array}{ccc}(S_{n,d}(\F),\dd_{\ess} )& \longrightarrow & (\F^{n \times \binom{n+d-2}{d-1}}, \dd_{\rk}) \\ 
f & \longmapsto & C_f 
\end{array}
\]
is an $\F$-linear  isometric embedding.
\end{corollary}

\begin{example}\label{exa:code_2dm}
Let $\C$ be the essential-rank-metric code given in Example \ref{exa:2}, that is $\C:=\langle f_1,f_2\rangle_{F}$, where 
\begin{align*}
f_1&=3 x_1^3 +  8 x_1^2 x_2 + 5 x_1^2x_3 + 12 x_1 x_2^2 + 4 x_1x_2 x_3 + 4 x_1x_3^2 + 8 x_2^3 + 2 x_2 x_3^2 + x_3^3,\\
f_2&=x_1^3+2x_1x_2^2+x_3^3.
\end{align*}
We already mentioned that $\C$ is a $[(3,3),2,2]_\F$ essential-rank-metric code. Here we show that the minimum essential-rank distance of $\C$ is in fact $2$. We use the embedding of Corollary \ref{cor:embedding} to see that $\C$ is isometric to the rank-metric code $\C':=\langle C_{f_1},C_{f_2}\rangle_{\F} \subseteq \F^{3 \times 6}$. The first catalecticant matrix $C_{f_1}$ of $f_1$ has been computed in Example \ref{exa:3}, while the first catalecticant matrix $C_{f_2}$ of $f_2$ is given by
\[
C_{f_2}=\begin{pmatrix}
 3 & 0 & 0 & 2 & 0 & 0  \\
 0 & 4 & 0 & 0 & 0 & 0 \\
 0 & 0 & 0 & 0 & 0 & 3
\end{pmatrix}.
\]
It is immediate to see that $\rk(C_{f_2})=3$. Furthermore, the $3$th and $5$th column of every matrix of the form $C_{f_1}+\lambda C_{f_2}$ with $\lambda\in\F$ are linearly independent. This shows that $\rk(A)\ge 2$ for every $A \in \C'$, and we can conclude that the minimum essential-rank distance of $\C$ is $2$.
\end{example}

\begin{remark}
Corollary \ref{cor:embedding} allows us to identify each essential-rank-metric code in $S_{n,d}(\F)$ with a rank-metric code in $\F^{n \times \binom{n+d-2}{d-1}}$. These identifications are compatible with equivalence: suppose that $\mathcal{C}_1,\mathcal{C}_2\subseteq S_{n,d}(\F)$ are two equivalent essential-rank-metric codes. This means that there exist   $A\in\GL(n,\F)$ and $\lambda \in \F^*$  such that $\mathcal{C}_2=\lambda\mathcal{C}_1\cdot A$. Let $S^{d-1}(A)$ be the matrix corresponding to the linear map $S_{n,d-1}(\F)\to S_{n,d-1}(\F), g\mapsto g(x\cdot A)$ given in \eqref{eq:equiv}  with respect to the basis of monomials. Then we have $C_{\lambda f\cdot A}=\lambda AC_f S^{d-1}(A)^{-1}$ for all $f\in S_{n,d}(\F)$ and hence 
\[
\{C_f\mid f\in\mathcal{C}_2\}=\lambda A\cdot \{C_f\mid f\in\mathcal{C}_1\}\cdot S^{d-1}(A)^{-1}\mbox{ and }\{C_f\mid f\in\mathcal{C}_1\}
\]
are  equivalent as rank-metric codes; see e.g. \cite[Definition 9]{sheekeysurvey}.
\end{remark}

Corollary \ref{cor:embedding} also allows us to derive bounds on the parameters of essential-rank-metric codes.

\begin{corollary}\label{cor:inherited_SBound}
Assume that $\mathrm{char}(\F)=0$ or $\mathrm{char}(\F)>d\ge 2$. Let $\C$ be an $[(n,d),k,r]_{\F}$ essential-rank-metric code. Then 
\[
k \le \binom{n+d-2}{d-1}(n-r+1).
\]
\end{corollary}
\begin{proof}
By Corollary~\ref{cor:embedding}, we know that $\C$ is isometric to the rank-metric code $\C'\subseteq \F^{n \times \binom{n+d-2}{d-1}}$. Hence the corollary follows from the Singleton bound on $\C'$ for rank-metric codes; see e.g.~\cite{delsarte1978bilinearforms}.
\end{proof}

\begin{remark}
The bound of Corollary~\ref{cor:inherited_SBound} is never better than the Singleton-like bound of Proposition~\ref{prop:singletonlike_bound}. Indeed, for fixed $n,d$, consider the function
\begin{eqnarray*}
F\colon \{1,\ldots,n\}&\to&\Z\\
r&\mapsto& \binom{n+d-2}{d-1}(n-r+1)-\binom{n+d-1}{d}+\binom{r+d-2}{d}.
\end{eqnarray*}
We have $F(n)=0$ as $\binom{a+1}{b}=\binom{a}{b}+\binom{a}{b-1}$ for $(a,b)=(n+d-2,d)$. And for $r\in\{1,\ldots,n-1\}$, we have
\begin{align*}
F(r)-F(r+1)&=\binom{n+d-2}{d-1}+\binom{r+d-2}{d}-\binom{r+d-1}{d}\\
&=\binom{n+d-2}{d-1}-\binom{r+d-2}{d-1}\geq 0
\end{align*}
as $\binom{a+1}{b}=\binom{a}{b}+\binom{a}{b-1}$ for $(a,b)=(r+d-2,d)$ and $r\leq n$. So $F(r)\geq0$ for all $r\leq n$.
\end{remark}

We now analyze the parameters  of  the codes $\C_{\rho}^{n,d}(\alpha)$ \textbf{as rank-metric codes}. 
By Proposition~\ref{prop:dimbound} and Corollary~\ref{cor:inherited_SBound}, we know that
\[
s(n,d,\rho)\leq\dim \C_{\rho}^{n,d}(\alpha)\leq t(n,d,\rho).
\]
for
\begin{align*}
s(n,d,\rho)&:=\binom{n+d-1}{d} -n\binom{\rho+d-3}{d-1},\\
t(n,d,\rho)&:=\binom{n+d-2}{d-1}(n-\rho+1).
\end{align*}
We now compare $s(n,d,\rho)$ and $t(n,d,\rho)$ asymptotically. First, when $d,\rho$ are fixed, we have
\[
\begin{array}{ccl}
s(n,d,\rho)&\sim&n^d/d!\\
t(n,d,\rho)&\sim&n^d/(d-1)!
\end{array}
\]
as $n$ grows large. We see that our codes (and the ambient spaces they live in) are a factor~$d$ away from having maximal dimension. When $d$ is fixed and $\rho=\tau n$ for some constant $0<\tau<1$, then
\[
\begin{array}{cccc}
s(n,d,\rho)&\sim&(1-d\tau^{d-1})&\cdot\, n^d/d!\\
t(n,d,\rho)&\sim&d(1-\tau)&\cdot\, n^d/d!
\end{array}
\]
When $\tau$ is small relative to $d$ (more precisely when $\tau^{-1}>d^{1/(d-2)}$), we have $d\tau^{d-1}<\tau$ and so our codes are again a factor $d$ away from attaining the upper bound.

\begin{remark}\label{rem:improving?}
The fact that the codes $\C_\rho^{n,d}(\alpha)$ cannot be optimal \emph{interpreted as rank-metric codes} is a direct consequence of the fact that the the bound of Proposition \ref{prop:singletonlike_bound} is 
a factor $d$ far from the one in Corollary \ref{cor:inherited_SBound}. Furthermore, we
also suspect the bound of Proposition \ref{prop:singletonlike_bound} to in fact not be tight. Indeed, when $d=2$, we will see later in Corollary~\ref{cor:symm_SBound} that the upper bound can be improved at least when $\F$ is a finite field. This result was proved in~\cite{schmidt2015symmetric} with an approach based on association schemes, which highly relies on the field $\F$ being finite. Whether this improved bound holds also when $\F$ is not finite is still an open question. However, over finite fields it would be interesting to check whether a similar approach could give an improved bound also for $d>2$.
\end{remark}

\begin{question}
    Is it possible to improve the Singleton-like bound of Proposition \ref{prop:singletonlike_bound}? 
\end{question}

Finally, we  show why the space of essential variables of a polynomial can be interpreted as a natural notion of \emph{support}. 
In coding theory, the theory of supports gives a useful combinatorial perspective on the elements of a metric space, and it allows one to relate concrete objects such as linear codes with more abstract objects such as matroids; see e.g. \cite{barg97}. In the case of rank-metric codes, that is, spaces of $n\times m$ matrices over a field $\F$, one can associate the study of the lattice of subspaces of $\F^n$. Concretely, the \textbf{(column) support} of $A=(a^1 \mid \cdots\mid a^m)\in\F^{n\times m}$ is the space 
$$\mathrm{supp}(A):=\langle a^1,\ldots,a^m\rangle \subseteq \F^n.$$ 
This notion has been studied in details in recent works; { see e.g. \cite{jurrius17}. }

\begin{proposition}
Assume that $\mathrm{char}(F)>d$ or $\mathrm{char}(\F)=0$ and let $f\in S_{n,d}(\F)$. Then
$$
V_{\ess}(f)=\mathrm{supp}(C_f)\cdot(x_1,\ldots,x_n):=\left\{v(x) \mid v \in \mathrm{supp}(C_f) \right\}.
$$
\end{proposition}
\begin{proof}
By Theorem \ref{thm:carlini_ess}, the space of essential variables of $f$ is spanned by 
$$
\{(\partial_1^{a_1}\ldots\partial_n^{a_n}) \circ f \,|\, a_1+\ldots+a_n=d-1\}=\{v_1(x),\ldots,v_m(x)\},
$$
where $m=\binom{n+d-2}{d-1}$ and the elements are ordered in lexicographic order. Putting the vectors $v_1,\ldots,v_m$ as rows of an $m\times n$ matrix, we obtain the $(d-1)$th catalecticant matrix $C_{f,d-1}$. We can conclude by observing that $(C_{f,d-1})^\top=DC_f$, for some invertible diagonal matrix $D\in\F^{m\times m}$; see e.g. \cite{iarrobino}. Thus, the span of the rows of $C_{f,d-1}$ coincides with the span of the columns of $C_f$, that is the support of $C_f$. 
\end{proof}

\subsection{The case of degree \texorpdfstring{$2$}{2}-polynomials}

We analyze in detail what happens when $d=2$ and $\mathrm{char}(\F)\neq2$. In this case, essential-rank-metric codes are isometric to spaces of symmetric matrices with the rank metric. 

\begin{proposition}\label{prop:identification_d=2}
For every $f\in S_{n,2}(\F)$, the matrix $C_f$ is symmetric. The map $f\mapsto C_f$ is an isometry between the space $S_{n,2}(\F)$ equipped with the essential rank metric and the space of $n\times n$ symmetric matrices over $\F$ equipped with the rank metric.
\end{proposition}
\begin{proof}
Writing 
\[
f=\sum_{1\leq i \leq j \leq n}a_{ij}x_ix_j,
\]
we see that
\[
C_f=\begin{pmatrix}
2a_{1,1} & a_{1,2} & \cdots & a_{1,n} \\
a_{1,2} & 2a_{2,2} & \ddots & \vdots \\
\vdots &\ddots & \ddots & a_{n-1,n}\\
a_{1,n} & \cdots &a_{n-1,n} &  2a_{n,n}
\end{pmatrix}
\]
is a symmetric matrix. The map $f\mapsto C_f$ is clearly a bijection between $S_{n,2}(\F)$ and the space of $n\times n$ symmetric matrices over $\F$ as $\mathrm{char}(\F)\neq2$. This bijection is in fact an isometry by Theorem~\ref{thm:carlini_ess}.
\end{proof}

In other words, essential-rank-metric codes for $d=2$ are in fact \textbf{symmetric rank-metric codes}. These were studied by several authors; see e.g. \cite{deboer1996,schmidt2015symmetric}. In this way, essential-rank-metric codes provide a generalization of symmetric rank-metric codes. \bigskip

When $d=2$ and the field $\F$ is finite, we can actually exploit finer bounds between $n,d,k$ and $r$. This is due to the work of Schmidt, who used a combinatorial approach based on association schemes \cite{schmidt2015symmetric}.

\begin{corollary}[{\cite[Theorem 3.3]{schmidt2015symmetric}}]\label{cor:symm_SBound}
Assume that $\F$ is finite and $\mathrm{char}(\F)> 2$. Let $\C$ be an $[(n,2),k,r]_{\F}$ essential-rank-metric code. Then 
\[
k \leq \begin{cases} \dfrac{n(n-r+2)}{2} & \mbox{ if } n-r \mbox{ is even, }\\[1em]
 \dfrac{(n+1)(n-r+1)}{2}  & \mbox{ if } n-r \mbox{ is odd. } \end{cases}
\]
\end{corollary}

 \begin{remark}
    Corollary \ref{cor:symm_SBound} is an improvement on the Singleton-like bound of Proposition~\ref{prop:singletonlike_bound}. Indeed, for $d=2$, the Singleton-like bound for an $[(n,2),k,r]_\F$ code reads as 
    $$ k \leq \frac{(n+r)(n-r+1)}{2},$$
    which is always worse, except for $r\in \{1,n\}$ where the two bounds coincide.
 \end{remark}

Next we show that, using the identification from Proposition~\ref{prop:identification_d=2}, the essential-rank-metric codes we construct in this paper for $d=2$ are in fact symmetric Gabidulin codes. To do so, we need to fix identifications between some vector spaces over $\F$.\bigskip

Fix a basis $\alpha_1,\ldots,\alpha_n$ of $\LL$ over $\F$. Let $\beta_1,\ldots,\beta_n$ be the dual basis of $\alpha_1,\ldots,\alpha_n$ with respect to the trace inner product and let $\gamma_1,\ldots,\gamma_n$ be a basis of $\LL$ over $\F$ such that 
\[
(\gamma_1,\ldots,\gamma_n)\perp(\sigma^i(\alpha_1),\ldots,\sigma^i(\alpha_n))
\]
for $i=0,\ldots,n-2$. Consider the following maps:
\[
\begin{array}{ccccccc}
\End_{\F}(\LL)&\to&\LL^n,&\phantom{space}&\LL^n&\to&\F^{n\times n},\\\bigskip
\varphi&\mapsto&(\varphi(\alpha_1),\ldots,\varphi(\alpha_n)),&&v&\mapsto&A_v,\\
\End_{\F}(\LL)&\to&\Bil(\LL\times\LL,\F),&&\Bil(\LL\times\LL,\F)&\to&\F^{n\times n},\\
\varphi&\mapsto&\left((x,y)\mapsto\Tr(\varphi(x)y)\right),&&f&\mapsto&B_f,
\end{array}
\]
where $v_i=\sum_{j=1}^n(A_v)_{ij}\beta_j$ for $i=1,\ldots,n$ and $f(\alpha_i,\alpha_j)=(B_f)_{ij}$ for $i,j=1,\ldots,n$. Each of these maps is an $\F$-linear isomorphisms and the following diagram commutes:
\[
\begin{array}{ccc}
\End_{\F}(\LL)&\rightarrow&\LL^n\\
\downarrow&&\downarrow\\
\Bil(\LL\times\LL,\F)&\rightarrow&\F^{n\times n}
\end{array}
\]
We first recall the construction of (regular) Gabidulin codes in each of these spaces. Let $\sigma$ be a generator of $\Gal(\LL/\F)$. Then a Gabidulin code is a subspace of the form
\[
\left\{\sum_{i=-k}^{\ell}\mu_i\sigma^i\,\middle|\, \mu_{-k},\ldots,\mu_{\ell}\in\LL\right\}\subseteq\End_{\F}(\LL)
\]
for $k,\ell\geq0$ with $k+\ell\leq n-2$. Inside $\Bil(\LL\times\LL,\F)$, this is
\[
\left\{(x,y)\mapsto\sum_{i=-k}^{\ell}\Tr(\mu_i\sigma^i(x)y)\,\middle|\, \mu_{-k},\ldots,\mu_{\ell}\in\LL\right\}.
\]\
Inside $\LL^n$, this is the code correpsonding to the generator matrix
\[
\begin{pmatrix}
\sigma^{-k}(\alpha_1)&\cdots&\sigma^{-k}(\alpha_n)\\
\vdots&&\vdots\\
\sigma^{\ell}(\alpha_1)&\cdots&\sigma^{\ell}(\alpha_n)
\end{pmatrix}
\]
and parity check matrix
\[
\begin{pmatrix}
\sigma^{-k}(\gamma_1)&\cdots&\sigma^{-k}(\gamma_n)\\
\vdots&&\vdots\\
\sigma^{-(n-2-\ell)}(\gamma_1)&\cdots&\sigma^{-(n-2-\ell)}(\gamma_n)
\end{pmatrix}
\]
over $\LL$, that is, the code $\sigma^{-k}( \mathcal G_{n-\ell-k,n}^\sigma(\alpha))$ as defined in \eqref{eq:Gabidulin}, where $\alpha=(\alpha_1,\ldots,\alpha_n)$.
And inside $\F^{n\times n}$, we get
\[
\left\{M\in\F^{n\times n}\,\middle|\,\sigma^i(\gamma)M\beta^\top=0\mbox{ for }i=-(n-2-\ell),\ldots,-k\right\},
\]
where $\gamma=(\gamma_1,\ldots,\gamma_n)$ and $\beta=(\beta_1,\ldots,\beta_n)$.\bigskip

Symmetric Gabidulin codes are the subspaces of (regular) Gabidulin codes consisting of all symmetric elements: We have a notion of transpose for elements in $\F^{n\times n}$, which induces notions of transpose for elements of $\End_{\F}(\LL)$, $\Bil(\LL\times\LL,\F)$ and $\LL^n$. 

\begin{lemma}
For $\mu\in\LL$ and $k\in\mathbb{Z}$, we have
\[
(\mu\sigma^k)^\top=\sigma^{-k}(\mu)\sigma^{-k}
\]
\end{lemma}
\begin{proof}
Let $A,B\in\F^{n\times n}$ be the matrices such that $\mu\sigma^k(\alpha_i)=\sum_{j=1}^nA_{ij}\beta_j$ for $i=1,\ldots,n$ and $\sigma^{-k}(\mu)\sigma^{-k}(\alpha_j)=\sum_{i=1}^nB_{ji}\beta_i$ for $j=1,\ldots,n$. Then we see that
$$
A_{ij}=\Tr(\mu\sigma^k(\alpha_i)\alpha_j)=\Tr(\sigma^{-k}(\mu)\alpha_i\sigma^{-k}(\alpha_j))=B_{ji}
$$
for $i,j=1,\ldots,n$. Hence $A^\top=B$ and so the lemma follows.
\end{proof}

So a symmetric Gabidulin code is a subspace of the form
\[
\left\{\mu_0\sigma^0+\sum_{i=1}^{\ell}(\mu_i\sigma^i+\sigma^{-i}(\mu_i)\sigma^{-i})\,\middle|\, \mu_0,\ldots,\mu_{\ell}\in\LL\right\}\subseteq\End_{\F}(\LL)
\]
The elements of such a code correspond in $\Bil(\LL\times\LL,\F)$ to bilinear forms of the form
\[
(x,y)\mapsto \Tr\!\left(\mu_0xy+\!\sum_{i=1}^{\ell}(\mu_i\sigma^i(x)y+\sigma^{-i}(\mu_i)\sigma^{-i}(x)y)\right)\!=\Tr\!\left(\mu_0xy+\!\sum_{i=1}^{\ell}\mu_i(\sigma^i(x)y+x\sigma^i(y))\right)
\]
for $\mu_0,\ldots,\mu_\ell\in\LL$. So these are indeed exactly the symmetric Gabidulin codes constructed in~\cite{schmidt2015symmetric}. Inside $\F^{n\times n}$, we get
\[
\left\{M\in\F^{n\times n}\,\middle|\,M^\top=M,\sigma^i(\gamma)M\beta^\top=0\mbox{ for }i=-(n-2-\ell),\ldots,-\ell\right\}.
\]
The following proposition shows that this is exactly the code we construct for $d=2$.
 
\begin{proposition}
The space
\[
\left\{M\in\F^{n\times n}\,\middle|\,M^\top=M,\sigma^i(\gamma)M\beta^\top=0\mbox{ for }i=-(n-2-\ell),\ldots,-\ell\right\}.
\]
equals the image of $\C_{n-2\ell}^{n,2}(\sigma^{n-2-\ell}(\beta),\gamma)$ under the isometry $f\mapsto C_f$.
\end{proposition}
\begin{proof}
As the map $f\mapsto C_f$ is an isometry between $S_{n,2}(\F)$ and the space over symmetric $n\times n$ matrices over $\F$, it suffices to check that, for $f\in S_{n,2}(\F)$, the conditions
\[
\sigma^{n-2-\ell}(\beta(\partial))\sigma^{j}(\gamma(\partial)) \circ f=0, \quad\quad j=0,\ldots,n-2\ell-2,
\]
and the conditions
\[
\sigma^{i}(\gamma)C_f\beta^\top=0,\quad\quad i=-(n-2-\ell),\ldots,-\ell,
\]
are equivalent. Take $j\in\{0,\ldots,n-2\ell-2\}$ and $i=j-(n-2-\ell)$. We prove that
\[
\sigma^{n-2-\ell}(\beta(\partial))\sigma^{j}(\gamma(\partial))f=0\Longleftrightarrow \sigma^{i}(\gamma)C_f\beta^\top=0
\]
for all $f\in S_{n,2}(\F)$. Both equations are linear in $f$ and the space $S_{n,2}(\F)$ is spanned by polynomials of the form $(v(x))^2$ with $v\in \F^n$. So we can assume that $f=(v(x))^2$ and $C_f=2v^\top v$. Now it is easy to check that
\begin{align*}
\sigma^{n-2-\ell}(\beta(\partial))\sigma^{j}(\gamma(\partial))\circ f&=\sigma^{n-2-\ell}\left(\beta(\partial)\sigma^{i}(\gamma(\partial)) \circ f\right)\\
&=\sigma^{n-2-\ell}\left(2(\beta\cdot v)(\sigma^i(\gamma)\cdot v)\right)\\
&=\sigma^{n-2-\ell}\left(2\sigma^i(\gamma)v^\top v\beta^\top\right)\\
&=\sigma^{n-2-\ell}\left(\sigma^i(\gamma)C_f\beta^\top\right)
\end{align*}
for all $v\in\F^n$ and so
\[
\sigma^{n-2-\ell}(\beta(\partial))\sigma^{j}(\gamma(\partial)) \circ f=\sigma^{n-2-\ell}\left(\sigma^{i}(\gamma)C_f\beta^\top\right)
\]
for all $f\in S_{n,2}(\F)$. This shows the claimed equivalence.
\end{proof}

\section{Other ranks of polynomials}\label{sec:other_ranks}

There are multiple natural notions of rank for homogeneous polynomials. Each of these ranks defines a measure of distance between polynomials, and for each one we can ask for contructions of high-dimensional codes with a large minimal distance. In this section, we consider the nested notions of rank defined below. Our goal in particular is to motivate the study of a new class of codes we call strength-metric codes.

\begin{definition}
Let $f \in S_{n,d}(\F)$. The \textbf{strength} of $f$ is the integer
\[
\str(f)\coloneqq \min\{ r \in \Z_{\geq0} \mid f=g_1h_1+\ldots+g_rh_r,~ g_i,h_i \in S_n(\F), \deg (g_i),\deg (h_i)<d \}.
\]
The \textbf{slice rank} of $f$ is the integer
\[
\ssl(f)\coloneqq \min\{ r \in \Z_{\geq0} \mid f=\ell_1g_1+\ldots+\ell_rg_r,~ \ell_i \in S_{n,1}(\F),g_i\in S_{n,d-1}(\F) \}.
\]
The \textbf{essential rank} of $f$ is the integer
\[
\ess(f)\coloneqq \min\{ r \in \Z_{\geq0} \mid f=g(\ell_1,\ldots,\ell_r),~ g\in S_{r,d}(\F),\ell_i\in S_{n,1}(\F)\}.
\]
The \textbf{Waring rank} of $f$ is the integer
\[
\war(f)\coloneqq \min \{ r \in \Z_{\geq0} \mid f=\lambda_1\ell_1^d+\ldots+\lambda_r\ell_r^d,~ \lambda_i \in \F,\ell_i \in S_{n,1}(\F)\}.
\]
\end{definition}

\begin{proposition}\label{prop:chain_ranks}
Let $f \in S_{n,d}(\F)$. Then
\[
\str(f)\leq\ssl(f)\leq\ess(f)\leq\war(f).
\]
\end{proposition}
\begin{proof}
A slice decomposition is in particular a strength decomposition. Hence $\str(f)\leq\ssl(f)$. When $f\in\F[\ell_1,\ldots,\ell_r]$, then $f$ has a slice decomposition with $r$ terms. Hence $\ssl(f)\leq\ess(f)$. When $f=\ell_1^d+\ldots+\ell_r^d$, then $f\in\F[\ell_1,\ldots,\ell_r]$. Hence $\ess(f)\leq\war(f)$.
\end{proof}

Proposition \ref{prop:chain_ranks} shows that codes with a high minimum strength distance also have a high minimum distance with respect to the slice rank, essential rank and Waring rank. The following result can be seen as a generalization of this statement.

\begin{theorem}[{\cite[Theorem 4]{bik2019strength}}]
Let $\F$ be a perfect field with $\mathrm{char}(\F)=0$ or $\mathrm{char}(\F)>d$ and let $\mathcal{P}$ be a property of homogeneous degree-$d$ polynomials over $\F$. Assume that for all $f\in S_{n,d}(\F)$ that have the property $\mathcal{P}$, $m\in\Z_{\geq0}$ and $\ell_1,\ldots,\ell_n\in S_{m,1}(\F)$, the polynomial $f(\ell_1,\ldots,\ell_n)\in S_{m,d}(\F)$ also has the property $\mathcal{P}$. Then either all polynomials have $\mathcal{P}$ or there exists a constant $k$ such that $\str(f)\leq k$ for all polynomials $f$ that have $\mathcal{P}$.
\end{theorem}

By letting $\mathcal{P}$ be the property of having rank $\leq\ell$ for some constant $\ell$ and one of the notions of rank above, we see that having high strength implies that the other ranks are also high. The theorem suggests that codes with a high minimum strength distance, since their elements do not have such properties, might have good decoding properties. All of this leads us to focus on the following question.

\begin{question}
Can one explicitly construct codes with a high minimal strength distance?
\end{question}

Already for codes of dimension $1$, i.e., codes spanned by a single polynomial, this question is nontrivial. In this case, the question asks for explicit polynomials of high strength. The following lemma shows one way to lower bound the strength of a polynomial.

\begin{lemma}\label{lm:strengthbound}
Let $f \in S_{n,d}(\F)$ and let $\Sing(f)$ be the singular locus of the hypersurface $\{f=0\}$. Then $\dim\Sing(f)\geq n-2\str(f)$. In particular, when $\Sing(f)=\{0\}$, we have $\str(f)\geq n/2$.
\end{lemma}
\begin{proof}
When $f=g_1h_1+\ldots+g_rh_r$ with $\deg g_i,\deg h_i<d$, then $\{g_1,h_1,\ldots,g_r,h_r=0\}$ is contained in the singular locus of $\{f=0\}$. This subvariety has codimension $\leq 2r$ in $\F^n$.
\end{proof}

Using Lemma \ref{lm:strengthbound}, we know how to construct polynomials with strength $\geq \lceil n/2\rceil$. The following theorems show, at least when $\mathrm{char}(\F)=0$, that this is not close to the maximal possible strength.

\begin{theorem}[{\cite[Theorem 12.8]{har1992alggoem}}]
Suppose that $\F$ is algebraically closed and $\mathrm{char}(\F)=0$. Then the generic slice rank of a polynomial in $S_{n,d}(\F)$ equals $\lceil n/2\rceil$ when $d=2$ and
\[
\min\left\{ r\in\Z_{\geq0} \,\middle|\, r(n+1-r) \geq \binom{n-r+d}{d}\right\}
\]
when $d\geq 3$.
\end{theorem}

\begin{theorem}[{\cite[Corollary 1.9]{bal2021genericstrength}}]
Suppose that $\F$ is algebraically closed and $\mathrm{char}(\F)=0$. Then the strength and slice rank of a generic element of $S_{n,d}(\F)$ coincide. 
\end{theorem}

As bounded slice rank is a Zariski-closed condition, the generic and maximal slice ranks and strengths all coincide. So we see that polynomials with much higher strength than $\lceil n/2\rceil$ exist. The best result known in this direction is a construction by Gesmundo-Ghosal-Ikenmeyer-Lysikov \cite{ges2022strengthdecomp} of polynomials with slice rank $\geq \lceil n/2\rceil +1$. Lemma~\ref{lm:strengthbound} however does give us a starting point for investigating codes with the strength metric.

\begin{question}
Can one explicitely construct high-dimensional codes in $S_{n,d}(\F)$ where every nonzero element defines a smooth projective hypersurface?
\end{question}

\appendix

\section{Orbits of small-width multisets modulo \texorpdfstring{$n$}{n}}

Fix integers $d\geq 2$ and $n>k\geq 0$. Consider the set $\mathcal{M}$ of $d$-element multisets 
\[
S=\{\!\!\{x_1,\ldots,x_d\}\!\!\}
\]
whose elements $x_1,\ldots,x_d$ lie in $\Z/n\Z$ and the action of $\Z$ on $\mathcal{M}$ given by
\[
x+S=\{\!\!\{x+y\mod n \mid y\in S\}\!\!\}
\]
We identify the elements of $\{0,\ldots,n-1\}$ with their images in $\Z/n\Z$ and consider the subset $\mathcal{R}$ of $\mathcal{M}$ consisting of all multisets whose elements are all contained in $\{0,\ldots,k\}$ and that contain $0$ at least once. We now have the following question.

\begin{question}
How many $\mathbb{Z}$-orbits of $\mathcal{M}$ intersect $\mathcal{R}$?
\end{question}

The goal of this appendix is to answer this question when $\frac{2}{3}n>k$. We first consider the case where $n\gg k$.

\begin{proposition}
Suppose that $\frac{1}{2}n>k$. Then every $S\in\mathcal{R}$ lies in a distinct $\mathbb{Z}$-orbit of $\mathcal{M}$.
\end{proposition}
\begin{proof}
Let $S,S'\in\mathcal{R}$ be distinct multisets and assume that $S'=x+S$ for some $x\in\Z$. We may assume that $x\in\{1,\ldots,n-1\}$. Note that $x=x+0\in x+S=S'$ and therefore $x\in\{0,\ldots,k\}$. We have $S=(n-x)+S'$. It follows that $n-x=(n-x)+0\in (n-x)+S'=S$ and hence $n-x\in\{0,\ldots,k\}$. So $n=x+(n-x)\leq k+k$. Contradiction, so $S,S'$ must be in distinct $\mathbb{Z}$-orbits.
\end{proof}

\begin{lemma}
The number of elements of $\mathcal{R}$ equals
\[
\binom{(d-1)+k}{d-1}.
\]
\end{lemma}
\begin{proof}
One way to encode a multiset $S\in\mathcal{R}$ is as a tuple $(n_0,\ldots,n_k)$ where $n_i$ equals the number of occurrences of $i\in\{0,\ldots,k\}$ in $S$. We then see that such tuples represent an element of $\mathcal{R}$ if and only if $n_0\in\N$, $n_1,\ldots,n_k\in\Z_{\geq0}$ and $n_0+\ldots+n_k=d$, and in this case the multiset $S$ is uniquely determined by $(n_0,\ldots,n_k)$. This shows that the size of $\mathcal{R}$ equals the number of $(k+1)$-tuples $(n_0-1,n_1,\ldots,n_k)$ of nonnegative integers adding up to $d-1$. There are
\[
\binom{(d-1)+k}{d-1}
\]
such tuples.
\end{proof}

\begin{proposition}
Suppose that $\frac{1}{2}n>k$. Then $\mathcal{R}$ intersects
\[
\binom{(d-1)+k}{d-1}
\]
$\mathbb{Z}$-orbits of $\mathcal{M}$.
\end{proposition}

The second case we consider is when $n$ is smaller.

\begin{proposition}
Suppose that $\frac{2}{3}n>k\geq \frac{1}{2}n$. Then the map $\pi$ sending
\[
((n_0,\ldots,n_a),(m_0,\ldots,m_b))\mapsto \bigcup_{i=0}^a\bigcup_{j=1}^{n_i}\{\!\!\{i\}\!\!\}\cup\bigcup_{i=0}^b\bigcup_{j=1}^{m_i}\{\!\!\{k-b+i\}\!\!\}
\]
is a bijection between the sets
\[
\left\{((n_0,\ldots,n_a),(m_0,\ldots,m_b))\,\middle|\,\begin{array}{l}a,b\in\Z_{\geq0}, n_0,m_0\in\N,n_1,\ldots,n_a,m_1,\ldots,m_b\in\Z_{\geq0},\\ a+b=2k-n,~ n_0+\ldots+n_a+m_0+\ldots+m_b=d,\\(n_0,\ldots,n_a)\neq (m_0,\ldots,m_b)\end{array}\right\}
\]
and $\left\{S\in\mathcal{R}\,\middle|\, \#(\Z\cdot S)\cap \mathcal{R}>1\right\}$. Furthermore, we have $\Z\cdot \pi((n_0,\ldots,n_a),(m_0,\ldots,m_b))\cap\mathcal{R}$ equals 
\[
\{\pi((n_0,\ldots,n_a),(m_0,\ldots,m_b)),\pi((m_0,\ldots,m_b),(n_0,\ldots,n_a))\}
\]
for all elements of the former set.
\end{proposition}
\begin{proof}
First consider $\pi$ as a map to $\mathcal{R}$. We start by showing that $\pi$ is injective. Clearly, it is enough to show that if
\[
S:=\pi((n_0,\ldots,n_a),(m_0,\ldots,m_b))=\pi((n'_0,\ldots,n'_{a'}),(m'_0,\ldots,m'_{b'})),
\]
then $(a',b')=(a,b)$. Suppose not. We may assume that $a'>a$ and $b'<b$. As $k-b\in S$ and $k-b<k-b'$, we see that $k-b\leq a'$. Hence $k\leq a'+b\leq (2k-n)+(2k-n)=4k-2n$ and so $2n\leq 3k$. This contradicts the assumption that $k<\frac{2}{3}n$. Hence $(a',b')=(a,b)$ and $\pi$ is injective. Next, we show that the image of $\pi$ lies in the stated codomain. Let $((n_0,\ldots,n_a),(m_0,\ldots,m_b))$ be an element of the domain. We show that the $\mathbb{Z}$-orbit of $\pi((n_0,\ldots,n_a),(m_0,\ldots,m_b))$ contains at least two elements of $\mathcal{R}$. This is indeed the case as
\[
\pi((m_0,\ldots,m_b),(n_0,\ldots,n_a))=n-k+b+\pi((n_0,\ldots,n_a),(m_0,\ldots,m_b))
\]
is also contained in this orbit and $(n_0,\ldots,n_a)\neq (m_0,\ldots,m_b)$.\bigskip

Let $S,S'\in\mathcal{R}$ be distinct multisets and assume that $S'=x+S$ for some $x\in\Z$. We may assume that $x\in\{1,\ldots,n-1\}$. Our goal is to prove that 
\[
S=\pi((n_0,\ldots,n_a),(m_0,\ldots,m_b))\mbox{ and }S'=\pi((m_0,\ldots,m_b),(n_0,\ldots,n_a))
\]
for some unique $n_0,\ldots,n_a,m_0,\ldots,m_b$ with the required properties.\bigskip

The first step is to compute $x$ from $S$. We have $x=x+0\in x+S=S'$ and $n-x=(n-x)+0\in (n-x)+S'=S$, so $x,n-x\in\{0,\ldots,k\}$. So $n-k\leq x\leq k$. As $0\in S'=x+S$ and $x\leq k$, we have $n-x\in S$ and $n-x\geq n-k$. So $n-x\geq  \min\{y\mid  y\in S,y\geq n-k\}$. We claim that equality holds. Suppose it does not. Then there is an $y\in S$ with $n-k\leq y<n-x$. We get the element
\[
z:=x+y\in x+S=S'
\]
which satisfies $n-k+x\leq z<n$. Since $z<n$ and $z\in S'$, we get $z\leq k$. So $n-k+x\leq z\leq k$. It follows that $x\leq 2k-n$ and so
\[
\frac{1}{3}n<n-k\leq x\leq 2k-n< 2\cdot\frac{2}{3}n-n=\frac{1}{3}n.
\]
Contradiction, so $n-x=\min\{y\mid y\in S,y\geq n-k\}$ must hold. We get
\[
x=n-\min\{y\mid y\in S,y\geq n-k\}.
\]
Define $b:=x-(n-k)\geq 0$ and $a:=(2k-n)-b=k-x\geq0$. We indeed have $a,b\in\Z_{\geq0}$ and $a+b=2k-n$. Next, we show that 
\[
S=\pi((n_0,\ldots,n_a),(m_0,\ldots,m_b))
\]
for some $n_0,n_1,\ldots,n_a,m_0,m_1,\ldots,m_b\in\Z_{\geq0}$. Let $z\in S$. Then we need to show that either $z\leq a$ or $z\geq k-b$ holds. If $z\geq n-k$, then $z\geq \min\{y\mid y\in S,y\geq n-k\}=n-x=k-b$. Suppose that $z<n-k$. Then $x+z\in x+S=S'$ and $x+z<x+(n-k)\leq x+(n-x)=n$. So $x+z\leq k$ must hold. So $z\leq k-x=a$. As every element of $S$ is contained in $\{0,\ldots,a\}\cup\{k-b,\ldots,k\}$, we get
\[
S=\pi((n_0,\ldots,n_a),(m_0,\ldots,m_b))
\]
where $n_i\in\Z_{\geq0}$ equals the number of occurrences of $i$ in $S$ and $m_j\in\Z_{\geq0}$ equals the number of occurrences of $k-b+j$ in $S$. As $0,n-x=k-b\in S$, we see that $n_0,m_0\in\N$. Finally, we have
\[
\pi((n_0,\ldots,n_a),(m_0,\ldots,m_b))=S\neq S'=\pi((m_0,\ldots,m_b),(n_0,\ldots,n_a))
\]
and therefore $(n_0,\ldots,n_a)\neq (m_0,\ldots,m_b)$. So $((n_0,\ldots,n_a),(m_0,\ldots,m_b))$ has the required properties and $S$ is contained in the image of $\pi$. So $\pi$ is surjective. As $x$ is uniquely determined by $S$, we see that the $\Z$-orbit of $S$ intersects $\mathcal{R}$ exactly in $\{S,S'\}$.
\end{proof}

\begin{example}
Suppose that $n=2k$. Then the multisets in $\mathcal{R}$ whose $\Z$-orbit intersects $\mathcal{R}$ twice are those of the form $\{\!\!\{0,\ldots,0,k\ldots,k\}\!\!\}$, where the numbers of $0$'s and $k$'s are $\geq 1$, distinct and add up to $d$.
\end{example}

We can now count the number of orbits intersecting $\mathcal{R}$ in this range.

\begin{theorem}
Suppose that $\frac{2}{3}n>k\geq \frac{1}{2}n$. Then $\mathcal{R}$ intersects
\[
\binom{(d-1)+k}{d-1}-\frac{1}{2}\left(\!(2k-n+1)\binom{(d-2)+(2k-n+1)}{d-2}\!-\delta_{2\mid d}\delta_{2\mid n}\binom{(d/2-1)+(k-n/2)}{d/2-1}\!\!\right)
\]
orbits of $\mathcal{M}$.
\end{theorem}
\begin{proof}
The proposition shows that an orbit of $\mathcal{M}$ intersects $\mathcal{R}$ in either $0$, $1$ or $2$ elements. So the number of orbits of $\mathcal{M}$ that intersect $\mathcal{R}$ equals the size of $\mathcal{R}$ minus half the size of the subset of $\mathcal{R}$ consisting of all elements $S$ such that $\#(\Z\cdot S)\cap \mathcal{R}>1$. The size of this subset of $\mathcal{R}$ equals
\[
\sum_{a=0}^{2k-n}\#\left\{((n_0,\ldots,n_a),(m_0,\ldots,m_{2k-n-a}))\,\middle|\!\begin{array}{l}n_0,m_0\in\N,n_1,\ldots,n_a,m_1,\ldots,m_{2k-n-a}\in\Z_{\geq0},\\ n_0+\ldots+n_a+m_0+\ldots+m_{2k-n-a}=d\end{array}\!\!\!\right\}
\]
minus
\[
\#\left\{((n_0,\ldots,n_a),(m_0,\ldots,m_b))\,\middle|\,\begin{array}{l}a,b\in\Z_{\geq0}, n_0,m_0\in\N,n_1,\ldots,n_a,m_1,\ldots,m_b\in\Z_{\geq0},\\ a+b=2k-n,~ n_0+\ldots+n_a+m_0+\ldots+m_b=d,\\(n_0,\ldots,n_a)=(m_0,\ldots,m_b)\end{array}\right\}.
\]
We have
\begin{eqnarray*}
&&\!\!\#\left\{\!((n_0,\ldots,n_a),(m_0,\ldots,m_{2k-n-a}))\,\middle|\!\begin{array}{l}n_0,m_0\in\N,n_1,\ldots,n_a,m_1,\ldots,m_{2k-n-a}\in\Z_{\geq0},\\ n_0+\ldots+n_a+m_0+\ldots+m_{2k-n-a}=d\end{array}\!\!\!\right\}\\
&=&\!\!\#\left\{((p_0+1,p_1,\ldots,p_a),(p_{a+1}+1,p_{a+2},\ldots,p_{2k-n+1}))\,\middle|\,\begin{array}{l}p_0,\ldots,p_{2k-n+1}\in\Z_{\geq0},\\ p_0+\ldots+p_{2k-n+1}=d-2\end{array}\right\}\\
&=&\!\!\binom{(d-2)+(2k-n+1)}{d-2}
\end{eqnarray*}
for all $a\in\{0,\ldots,2k-n\}$ and
\begin{eqnarray*}
&&\!\#\left\{\!((n_0,\ldots,n_a),(m_0,\ldots,m_b))\,\middle|\!\begin{array}{l}a,b\in\Z_{\geq0}, n_0,m_0\in\N,n_1,\ldots,n_a,m_1,\ldots,m_b\in\Z_{\geq0},\\ a+b=2k-n,~ n_0+\ldots+n_a+m_0+\ldots+m_b=d,\\(n_0,\ldots,n_a)=(m_0,\ldots,m_b)\end{array}\!\!\!\right\}\\
&=&\!\delta_{2\mid n}\cdot\#\left\{((n_0,\ldots,n_{k-n/2}),(n_0,\ldots,n_{k-n/2}))\,\middle|\,\begin{array}{l}n_0\in\N,n_1,\ldots,n_{k-n/2}\in\Z_{\geq0},\\ 2(n_0+\ldots+n_{k-n/2})=d\end{array}\right\}\\
&=&\!\delta_{2\mid d}\delta_{2\mid n}\cdot\#\left\{(n_0'+1,n_1,\ldots,n_{k-n/2})\,\middle|\,\begin{array}{l}n_0',n_1,\ldots,n_{k-n/2}\in\Z_{\geq0},\\ n_0'+\ldots+n_{k-n/2}=d/2-1\end{array}\right\}\\
&=&\!\delta_{2\mid d}\delta_{2\mid n}\binom{(d/2-1)+(k-n/2)}{d/2-1}.
\end{eqnarray*}
This yields the stated equality.
\end{proof}

\end{document}